\pdfoutput=1

\documentclass{article}

\usepackage{arxiv}

\usepackage[utf8]{inputenc} 
\usepackage[T1]{fontenc}    
\usepackage{url}            
\usepackage{booktabs}       
\usepackage{amsfonts}       
\usepackage{nicefrac}       
\usepackage{microtype}      

\usepackage[autostyle=true]{csquotes} 

\usepackage{bm}
\usepackage{bbm}
\usepackage{graphics}
\usepackage{hyperref}
\usepackage{dsfont}
\usepackage{enumitem}
\usepackage{amsmath}
\usepackage{amsfonts}
\usepackage{amssymb}
\usepackage{mathtools}
\usepackage{subcaption}
\usepackage[table]{xcolor}
\usepackage{tcolorbox}
\usepackage{pdfpages}
\usepackage[standard]{ntheorem}
\usepackage{microtype}
\usepackage[ruled]{algorithm2e}
\usepackage{array}
\RequirePackage{scrextend}
\changefontsizes[14]{11}

\RequirePackage{tabularx}
\usepackage[natbibapa]{apacite}

\usepackage{natbib}
\renewcommand\cite{\citep}

\newcolumntype{H}{>{\setbox0=\hbox\bgroup}c<{\egroup}@{}}
\newcolumntype{L}[1]{>{\raggedright\let\newline\\\arraybackslash\hspace{0pt}}m{#1}}
\newcolumntype{C}[1]{>{\centering\let\newline\\\arraybackslash\hspace{0pt}}m{#1}}
\newcolumntype{R}[1]{>{\raggedleft\let\newline\\\arraybackslash\hspace{0pt}}m{#1}}

\mathtoolsset{
	showonlyrefs,
}

\title{A multiple testing framework for diagnostic accuracy studies with co-primary endpoints}

\author{
  Max Westphal\thanks{Correspondence to: Max Westphal,\ \ \href{mailto:mwestphal@uni-bremen.de}{mwestphal@uni-bremen.de},\ \ \url{https://orcid.org/0000-0002-8488-758X}}\\
  Institute for Statistics\\
  University of Bremen\\
  Bremen, Germany \\
   \And
 Antonia~Zapf \\
  Institute of Medical Biometry\\ and Epidemiology\\
  UKE Hamburg\\
  Hamburg, Germany\\
   \And
  Werner~Brannath \\
  Institute for Statistics and\\ 
  Competence Center for\\
  Clinical Trials Bremen\\
  University of Bremen\\
  Bremen, Germany \\
}

\date{March 22, 2020}

\DeclareMathOperator{\argmax}{argmax}

\DeclareMathOperator{\cov}{cov}

\DeclareMathOperator{\FWER}{FWER}

\DeclareMathOperator{\diag}{diag}

\DeclareMathOperator{\Acc}{Acc}
\DeclareMathOperator{\bAcc}{bAcc}

\DeclareMathOperator{\se}{se}

\DeclareMathOperator{\CIbm}{\mathbf{CI}}
\DeclareMathOperator{\Bin}{Bin}

\DeclareMathOperator{\Se}{Se}
\DeclareMathOperator{\Sp}{Sp}
\DeclareMathOperator{\EFP}{EFP}
\DeclareMathOperator{\Beta}{Beta}
\DeclareMathOperator{\mBeta}{mBeta}
\DeclareMathOperator{\mBin}{mBin}

\DeclareMathOperator{\MAE}{MAE}
\DeclareMathOperator{\pmin}{pmin}

\newcommand{\pr}{\mathbb{P}}
\newcommand{\E}{\mathbb{E}}

\newcommand{\hatf}{\hat{f}}

\newcommand{\one}{\mathds{1}}
\newcommand{\setM}{\mathcal{M}}
\newcommand{\setN}{\mathcal{N}}
\newcommand{\setS}{\mathcal{S}}
\newcommand{\setT}{\mathcal{T}}
\newcommand{\setL}{\mathcal{L}}
\newcommand{\setE}{\mathcal{E}}

\newcommand{\thetahat}{\hat{\theta}}

\newcommand{\setP}{\mathcal{P}}

\newcommand{\setH}{\mathcal{H}}
\newcommand{\setV}{\mathcal{V}}

\newcommand{\D}{\mathfrak{D}}
\newcommand{\Sehat}{\widehat{\Se}}
\newcommand{\Sphat}{\widehat{\Sp}}
\newcommand{\Sehatbm}{\widehat{\Sebm}}
\newcommand{\Sphatbm}{\widehat{\Spbm}}

\newcommand{\CISebm}{\CIbm^{\Se}}
\newcommand{\CISpbm}{\CIbm^{\Sp}}
\newcommand{\TSe}{T^{\Se}}
\newcommand{\TSp}{T^{\Sp}}
\newcommand{\TSebm}{\Tbm^{\Se}}
\newcommand{\TSpbm}{\Tbm^{\Sp}}

\newcommand{\Sebm}{\mathbf{Se}}
\newcommand{\Spbm}{\mathbf{Sp}}
\newcommand{\Tbm}{\bm{T}}
\newcommand{\bb}{\bm{b}}
\newcommand{\bbhat}{\hat{\bb}}
\newcommand{\BBhat}{\hat{\BB}}

\newcommand{\BB}{\bm{B}}

\newcommand{\tra}{^\top}
\newcommand{\given}{\, |\,}
\newcommand{\EFPhat}{\widehat{\EFP}}
\newcommand{\EFPbm}{\mathbf{EFP}}
\newcommand{\EFPbmhat}{\widehat {\EFPbm} }
\newcommand{\op}{{\bm \oplus}}

\newcommand{\p}{{p}}

\mathtoolsset{
	showonlyrefs,
}

\hypersetup{pdftitle={A multiple testing framework for diagnostic accuracy studies with co-primary endpoints}} 
\hypersetup{pdfauthor={Max Westphal, Antonia Zapf, Werner Brannath}} 
\hypersetup{pdfkeywords={artificial intelligence, machine learning, medical device, medical testing, model selection}} 
\hypersetup{
		colorlinks,
		breaklinks=true,
		linkcolor={red!50!black},
		citecolor={blue!50!black},
		urlcolor={red!50!black}
}

\begin{document}

\maketitle

\begin{abstract}
Major advances have been made regarding the utilization of artificial intelligence in health care. In particular, deep learning approaches have been successfully applied for automated and assisted disease diagnosis and prognosis based on complex and high-dimensional data. However, despite all justified enthusiasm, overoptimistic assessments of predictive performance are still common. Automated medical testing devices based on machine-learned prediction models should thus undergo a throughout evaluation before being implemented into clinical practice. In this work, we propose a multiple testing framework for (comparative) phase III diagnostic accuracy studies with sensitivity and specificity as co-primary endpoints. Our approach challenges the frequent recommendation to strictly separate model selection and evaluation, i.e. to only assess a single diagnostic model in the evaluation study. We show that our parametric simultaneous test procedure asymptotically allows strong control of the family-wise error rate. Moreover, we demonstrate in extensive simulation studies that our multiple testing strategy on average leads to a better final diagnostic model and increased statistical power. To plan such studies, we propose a Bayesian approach to determine the optimal number of models to evaluate. For this purpose, our algorithm optimizes the expected final model performance given previous (hold-out) data from the model development phase. We conclude that an assessment of multiple promising diagnostic models in the same evaluation study has several advantages when suitable adjustments for multiple comparisons are implemented. 
\end{abstract}

\section{Introduction}\label{sec:introduction}
Research projects concerned with the application of machine learning techniques for disease diagnosis and prognosis have steadily grown in number over the last years. This is indicated, among others, by several review and overview publications \cite{jiang2017, litjens2017, miotto2017, ching2018}. In particular, the capabilities of end-to-end deep learning approaches on such supervised learning tasks are highly promising. For instance, vast advances have been reported in the literature regarding cancer diagnosis with deep neural networks \cite{hu2018}. End-to-end deep learning refers to a trend involving deep (neural network) model architectures which are able to learn highly complex relationships between predictors and the target variable while having less parameters than traditional (more shallow) models with comparable performance \cite{DL}.
In the training process, highly complex features are derived automatically by the learning algorithm \cite{lecun2015}. This framework contrasts the traditional pipeline of domain specific data preprocessing and hand-crafted features in combination with simpler prediction models.
Despite all the recent success of machine learning, there are still challenges regarding over-optimistic conclusions drawn from finite datasets which may to a large extend be attributed to the following two (broad) categories:
\begin{enumerate}
	\item \textbf{Study design and reporting:} The most popular recommendation to split data for training, selection and evaluation is frequently employed in practice \cite{APM, ESL, ELA, EMLM, DL, HOML}. In the machine learning community, the according datasets are commonly denoted as training, validation and test set. However, the less of the defining properties of the study (e.g. data source, data splitting modalities, performance measure, comparator, etc.) are specified in advance, the more 'opportunities' present themselves to the researchers to influence the results to match their expectations \cite{boulesteix2009, boulesteix2009optimal, jelizarow2010, boulesteix2013}. For diagnostic accuracy studies, several design-related sources of bias have been identified in the literature \cite{lijmer1999, whiting2004, rutjes2006, whiting2011, schmidt2013}. According to \citet{rutjes2006}, relevant sources for overoptimistic conclusions are
	selection of non-consecutive patients,
	analysis of retrospective data
	and focus on severe cases and healthy controls. 
	Another issue is an in-transparent communication (conscious or not) of the results \cite{ochodo2013, TRIPOD, STARD}. This may prevent other research teams to replicate the findings on similar problems (replicability) or even on the same data (reproducibility) \cite{jasny2011, peng2011}. 
	\item \textbf{Sampling variability}: The true predictive performance of a diagnostic model in the population of interest is not known but only estimated based on data. Even in the ideal case in which the evaluation data is independent of the data used for model development and all modalities of the evaluation study are specified in advance in a study protocol, the empirical performance is still a random variable and may realize at a high value just due to chance. From our perspective, this fact is often times overlooked within the machine learning community. 
\end{enumerate}

The primary goal of this work is a better utilization of the available data by optimizing study design and statistical analysis without introducing overoptimism. Our primary statistical inference goal is to bound the probability for false positive claims regarding the performance of a prediction model, i.e. the type I error rate. On the other hand, we of course seek to identify a truly good model with high probability. 
It may be argued that methodological research in machine learning or bioinformatics does not (need to) have the goal to control the type I error rate. We agree with this view in the context of the early development of new algorithms and prediction models and acknowledge that exploratory research is of utmost importance for the scientific process. In machine learning applications, this amounts to trying out a wide variety of learning algorithms deemed suitable for the prediction task at hand. This makes sense, as by the no-free-lunch theorem of statistical learning, no single algorithm gives universally best results for all prediction tasks \cite{UML}. Modern algorithms additionally involve the tuning of several hyperparameters. The performance of (deep) neural networks depends for example on depth (number of layers), width (neurons per layer), activation function(s) and further hyperparameters related to regularization (e.g. dropout rate) and optimization (e.g. loss function, learning rate, mini batch size) \cite{bengio2012}. In effect, usually dozens, hundreds or even thousands of candidate models are initially trained and compared.

However, at some point a reliable statement regarding the model performance is required. This is in particular the case in medical applications where the consequences of a wrong decision could ultimately be life-threatening, for instance when a flawed (automated) diagnostic tool is implemented into clinical practice. Additionally, depending on the regulatory context, a throughout evaluation study might even be mandatory. \citet{pesapane2018} give an overview over the current regulatory framework for artificial intelligence empowered medical devices in the EU and the US from the viewpoint of radiology. They also comment on future developments, in particular the new Medical Devices Regulation (MDR) and the new In Vitro Diagnostic Medical Device Regulation (IVDR) in the EU \citep{MDR, IVDR} which will certainly be relevant for the application of machine learning methods for medical testing purposes.

Throughout investigations of devices for automated medical testing purposes are still rather scarce (as are such devices themselves). A recent positive example covers the IDx-DR device which allows diagnosis of diabetic retinopathy in diabetic patients via deep learning based on retinal color images \cite{abramoff2016}. Diagnostic accuracy was assessed in an extensive observational clinical trial involving 900 patients, which lead to a marketing permit by the FDA\footnote{\url{https://www.accessdata.fda.gov/cdrh_docs/reviews/DEN180001.pdf} (accessed March 22, 2020)}. In this study, only a single model was (successfully) evaluated on the final dataset. This default strategy is often advised in machine learning and diagnostic accuracy research and is reasonable if previously available data for model training and selection is (a) large in number and (b) representative of the intended target population. However, when potentially hundreds of modeling approaches are compared on (quantitatively and/or qualitatively) modest datasets, the model selection process can rarely be concluded with confidence. While the default approach enables an unbiased estimation and simple statistical inference, one is thus bound to this one-time model choice under uncertainty. In effect, this strategy is quite inflexible as it is impossible to retrospectively correct a flawed model selection without compromising the statistical inference in the evaluation study.

To address this issue, \citet{EOMPM1, IMS} recently adapted a multiple testing approach from \citet{hothorn2008} to explicitly take into account that multiple models are assessed simultaneously on the same evaluation dataset. In effect, model selection can be improved with help of the test data \citep{IMS}. The employed simultaneous test procedure is based on the (approximate) multivariate normal distribution of performance estimates. An advantage of this procedure is that the multiplicity adjustment needs to be less strict when candidate models give (highly) similar predictions. This approach allows approximate control of the overall type I error rate and construction of simultaneous confidence regions as well as corrected (median-conservative) point estimates. Moreover, it was found that selecting multiple promising models can increase statistical power for the evaluation study compared to the default approach where only the best (cross-)validation model is evaluated. The main goal of this work is to extend this existing framework to diagnostic accuracy trials with co-primary endpoints sensitivity and specificity. 

In this context, a concise connection to the taxonomy of diagnostic research seems appropriate. Among the many resources available regarding this topic, terminology varies quite substantially \cite{pepe2003, pepe2008, zhou2009, CD}. We will adhere to \citet[chapter 2]{CD}, who established a system of five study phases in diagnostic research. Our work is mainly tailored towards cross-sectional phase III (comparative) diagnostic test accuracy (DTA) studies which ask the question:
\begin{quote}
	"Among patients in whom it is clinically sensible to suspect the target disorder, does the level of the test result distinguish those with and without the target disorder?" \newline \phantom{.} \hfill \cite{CD}
\end{quote}
The goal of a such studies is generally to provide evidence that a new index (or candidate) diagnostic test outperforms either a given comparator or, if no such competing test is available, a given performance threshold. Performance, i.e. sensitivity and specificity, is measured with regard to the (defined) ground truth which is in the optimal case derived by a so-called gold standard. It may however be infeasible to implement such a gold standard in the DTA study or there may not even exist one. A reference standard is thus typically defined as the best available approximation of the gold standard. Common reasons to install an index test in clinical practice are either lower invasiveness or costs. This balances the fact that, by definition, the index test can not have better performance than the reference standard. This is relevant for predictive modeling as, from the model evaluation perspective, a machine-learned prediction model is nothing more than a (potentially very complex) diagnostic test - at least when embedded in a medical device used for automated or assisted disease diagnosis. We refer to \citet[chapters 1-3]{CD}
for a general introduction in diagnostic research and a throughout treatment of DTA studies. Study design and reporting issues are discussed by \citet{ochodo2013}, \citet{TRIPOD} and \citet{STARD}.

Multiple testing methodology is not particularly popular, neither for medical test evaluation nor in machine learning \cite{IMS}. As indicated above, the recommendation to evaluate a single final model or test on independent data has prevailed in both domains. In this case, no adjustments for multiplicity are necessary as the final model is selected independently of (i.e. prior to) the evaluation study.  
Our approach has some overlap to so-called benchmark experiments in predictive modeling. Among others, \citet{hothorn2005} and \citet{demvsar2006} showcase different approaches for the comparison of learning algorithms over multiple (real or artificial) datasets and inference regarding their expected performances. In contrast, we aim to evaluate prediction models (conditional on the learning data), which is arguably more important at the end of the model building process - right before implementation of a specific model in (clinical) practice. In medical testing applications, it is natural to consider only a single index (candidate) test. This is particularly the case, when different tests are based on separate biological samples. It would then often not be considered ethically justifiable to assess more than a single index test and the reference test on the same patients. In contrast, when all tests are derived from the same data sample, e.g. retinal images, the number of index tests (diagnostic models) should be primarily guided by statistical considerations as long as the type I error rate can be controlled.

The primary goal of this work is to provide a multiple testing framework for the assessment of sensitivity and specificity as co-primary endpoints for multiple diagnostic procedures (prediction models) on the same data. To our knowledge, such an approach has not been proposed in the literature so far. In section \ref{sec:model}, we introduce core notation and assumptions and establish our inference framework. In section \ref{sec:practical}, we describe a novel Bayesian approach to determine the optimal number of models to include in the final evaluation study based on preliminary data. 
In section \ref{sec:experiments}, we show results of several numerical experiments: Our main goals are (a) an assessment of finite sample properties regarding (e.g.) control of the type I error rate and (b) a comparison of different approaches for model selection prior to the evaluation study. 
Finally, in section \ref{sec:discussion}, we summarize our findings, point out limitations of our framework and give an outlook on possible extensions.

\section{Statistical model}\label{sec:model}

\subsection{Prerequisites}\label{sec:assumptions}

We consider predicting a binary label $Y \in \{0,1\}$ based on $P \in \mathbb{N}$ features $\bm X \in \mathbb{R}^P$. Our canonical example in this work will be the distinction between diseased (having a given target condition, $Y=1$) and healthy ($Y=0$) subjects. In medical applications, the target $Y$ should be obtained by the (best available) reference standard and forms the ground truth for learning and evaluation. The goal of supervised machine learning is to provide a prediction model $\hatf: \bm x \mapsto \hat{y}$ with high performance. In practice, this is achieved by a learning algorithm $A$ (e.g. stochastic gradient descent) which outputs a prediction model $\hatf$ (e.g. a neural network) based on learning data $\setL = \{(\bm x_i, y_i)\}_{i=1}^{n_\setL}$. We write $\hatf=A(\setL)$ for short. We emphasize that we are concerned with the assessment of a specific model $\hatf$ trained on specific data $\setL$ and not with the properties of the algorithm $A$ that was used to learn $\hatf$. Finally, we note that, from the viewpoint of predictive modeling, disease diagnosis and prognosis essentially only differ regarding the time lag between capturing the features $\bm x$ and the label $\bm y$ - at least if this lag is approximately equal in the population of interest. Otherwise, a time to event analysis may be more appropriate. For the sake of brevity, we will focus on diagnosis tasks in the following.
In medical diagnosis, sensitivity ($\Se$) and specificity are ($\Sp$) are often both assessed simultaneously. We thus consider the tuple $\theta=(\Se, \Sp)$ as the performance measure of interest where  
\begin{align}\label{eq:sesp}
\Se = \pr (\hatf(\bm X)=1 \given Y=1) \quad \text{and} \quad \Sp = \pr (\hatf(\bm X)=0 \given Y=0).
\end{align}
This is advantageous compared to just considering the overall accuracy $\Acc = \pr(\hatf(X)=Y) = {\varrho \Se + (1-\varrho) \Sp}$ in the sense that neither the accuracy in healthy or diseased alone dominates our perception of the predictive performance if the disease prevalence $\varrho = \pr(Y=1)$ is small.

As we have outlined in section \ref{sec:introduction}, it is usually recommended to conduct a final performance assessment on independent evaluation or test data $\setE= \{(\bm x_i, y_i)\}_{i=1}^{n_\setE}$. Performance estimates on the learning data are generally a bad indicator of the true performance, in particular when a highly complex model might have overfitted the learning data. We assume that $\setE$ is an i.i.d. sample from the unknown joint distribution $\D_{(\bm X, Y)}$ of $\bm X$ and $Y$.  
Moreover, we assume that $M \in \mathbb{N}$ models $\hatf_m=A_m(\setL),\ m \in \setM=\{1,\ldots,M\}$, have been initially trained in the learning phase via different learning algorithms $A_m$. A subset $\setS \subset \setM$ of these initial candidate models is selected for evaluation. To simplify the notation we assume that the candidate models are ordered such that $\setS= \{1,\ldots,S\}$ with $1\leq S\leq M$. Different approaches to identify $\setS$ are discussed in section \ref{sec:planning}. This is implemented in practice by training models first only on the so-called training data $\setT \subset \setL$, a subset of the learning data. The remaining validation data $\setV = \setL \setminus \setV$ is than used for a performance assessment of the resulting preliminary models $\hatf^-_m=A_m(\setT)$. The model(s) selected for evaluation are than re-trained with all available data $\setL= \setT \cup \setV$ before going into the evaluation study with the expectation that this (slightly) increases their performance. This procedure is know as (simple) hold-out validation. There exist several variations and alternative strategies such as cross-validation or bootstrapping \cite{APM, DL}. They can stabilize model selection, are however computationally more expensive as models have to be re-trained several times which is not always feasible in practice. We will thus focus on simple hold-out validation in the numerical experiments in this work.

Note that by performance, we refer to the true population performance $\theta=(\Se, \Sp)$ of a model $\hatf$ as defined in \eqref{eq:sesp} for sensitivity and specificity. Validation performance is referring to the empirical performance of the preliminary model $\hatf^-$ (trained only on $\setT$) estimated on the validation data and denoted by $\thetahat(\setV)=(\Sehat(\setV)$, $\Sphat(\setV))$, in slight abuse of notation. In contrast, evaluation or test performance is referring to the empirical performance of the final model $\hatf$ (re-trained on $\setL$). It is estimated on the final evaluation or test dataset $\setE$ and denoted as $\thetahat=(\Sehat, \Sphat)$. We will often deal with several models $\hatf_m$, $m\in \setM$, and will usually refer to a specific model $\hatf_m$ just by the corresponding index $m$.

\subsection{Study goal}\label{sec: goals}

Our overall goal is to identify a diagnostic or prognostic prediction model with high sensitivity and specificity from the $M$ initially trained models. Additionally, we aim to show superiority of at least one model compared to prespecified thresholds $\theta_0=(\Se_0, \Sp_0) \in (0,1)^2$. 
Alternatively, $\theta_0 =\theta(\hatf_0)$ is the unknown performance of a comparator $\hatf_0$, i.e. an established diagnostic testing procedure which is also estimated in the evaluation study. In the following, we will focus on the former scenario which is also simpler to implement in our numerical experiments in section \ref{sec:experiments}. Throughout this work, we assume that the evaluation study is declared as successful if and only if superiority in both endpoints for at least one candidate model $m \in \setS$ relative to the the comparator $\theta_0$ can be demonstrated. This so-called co-primary endpoint analysis is the standard approach in confirmatory diagnostic accuracy studies \citep{committee2009}. \citet{vach2012} discuss other approaches in detail. The system of null hypotheses is thus given by
\begin{align}\label{eq:hyp}
\setH_\setS = \{H_m: H_m^{\Se} &\cup H_m^{\Sp},\ m\in \setS \}\\
\text{where} \quad H_m^{\Se}: \Se_m \leq \Se_0 \quad &\text{and}\quad H_m^{\Sp}: \Sp_m \leq \Sp_0.
\end{align}
Equivalently, the hypothesis system can be expressed as
\begin{align}\label{eq:hyp_alt}
\setH_\setS = \{H_m: \vartheta_m = \min(\Se_m, \Sp_m +\Delta_0) \leq \vartheta_0,\ m\in \setS \}.
\end{align}
Hereby, $\Delta_0=\Se_0-\Sp_0 \in (-1,1)$ specifies to which extent we prioritize sensitivity over specificity (or vice versa). For simplicity, we only consider $\Delta_0=0$ in the numerical experiments in this work.
The goal of the evaluation study is to estimate the unknown sensitivities and specificities, provide confidence bounds for these estimates and obtain a multiple test decision for the hypothesis system \eqref{eq:hyp}.

A multiple test $\bm \varphi = (\varphi_1, \ldots,\varphi_S)$ depends on the evaluation data $\setE$ and results in a binary vector ${\bm \varphi \in \{0,1\}^S}$ of test decisions whereby $H_m$ is rejected if and only if $\varphi_m =1$. In a regulatory setting, it is usually required that $\bm \varphi$ shall control the family-wise error rate (FWER) strongly at significance level $\alpha$, i.e. we require for all possible parameter configurations $\bm \theta$
\begin{align}\label{FWERcontrol}
\FWER_{\bm \theta}(\bm \varphi) = \pr_{ \bm \theta}\left( \bigcup_{m \in \setS_0} \{\varphi_{m} = 1\}\right) \leq \alpha.
\end{align}
Hereby, $\setS_0=\setS_0(\bm \theta, \theta_0) \subset \setS$ is the index set of true null hypothesis, i.e. we have either $\Se_m \leq \Se_0$ or $\Sp_m \leq \Sp_0$ for all $m \in \setS_0$. Any parameter $\bm \theta$ for which $\FWER_{\bm \theta}(\bm \varphi)$ becomes maximal is called a least favorable parameter configuration (LFC). Besides $\Sebm$ and $\Spbm$, the FWER in \eqref{FWERcontrol} also depends on the dependency structures of $\Sehatbm$ and $\Sphatbm$. For conciseness, we will however stick to the incomplete notation $\bm \theta =(\Sebm, \Spbm)$. In this work, we will restrict our attention to asymptotic control of the FWER, i.e. \eqref{FWERcontrol} shall only hold as $n=n_\setE \rightarrow \infty$. We will investigate the finite sample FWER in realistic and least-favorable settings in section \ref{sec:experiments}. As described by \citet{EOMPM1}, we extend $\bm \varphi$ to a multiple test for the actually relevant hypothesis system 
\begin{align}\label{eq:hyp_ext}
\setH = \setH_\setM = \{H_m: \vartheta_m \leq \vartheta_0,\ m\in \setM \}
\end{align}
concerning all initial candidate models $\hatf_m$, $m \in \setM$, by setting $\varphi_m = 0$ for all $m\in \setM\setminus \setS$. That is to say, a model $m$ cannot be positively evaluated ($\varphi_m=1$) when it is not selected for evaluation. This natural definition has two consequences. Firstly, the extended test retains (asymptotic) FWER control as only non-rejections are added. Secondly, we can compare different model selection strategies because the extended multiple test always operates on $\setH=\setH_\setM$ and not only on $\setH_\setS$ \citep{EOMPM1}.

\subsection{Parameter estimation}\label{sec:estimation}

We assume that a subset of promising models $\setS = \{1,\ldots,S\} \subset \setM$ has been selected prior to the evaluation study. Suitable strategies for that matter are presented in section \ref{sec:planning}. The observed feature-label data $\setE = \{(\bm x_i, y_i)\}_{i=1}^{n}$ from the evaluation study is transformed to the actual relevant binary similarity matrices $\bm Q^{\Se} \in \{0,1\}^{n_1 \times S}$ and $\bm Q^{\Sp} \in \{0,1\}^{n_0 \times S}$ for the diseased ($y=1$) and healthy ($y=0$) subpopulation, respectively. Hereby, $n_1$ and $n_0$ are the number of diseased and healthy subjects of the $n=n_\setE=n_1 + n_0$ evaluation study subjects. The entry of $\bm Q$ in row $i$ and column $m$ is equal to one if the prediction of the $i$-th observation by the $m$-th model is correct and zero if it is wrong. The sample averages
\begin{align}
\Sehat=\frac{1}{n_1}\sum_{i=1}^{n} \one(\hatf(\bm x_i)=y_i=1) \quad \text{and} \quad \Sphat=\frac{1}{n_0}\sum_{i=1}^{n} \one(\hatf(\bm x_i)=y_i=0).
\end{align}
can thus be calculated as the column means of these similarity matrices \cite{EOMPM1, IMS}.
Moreover, we can estimate the covariances $\bm \Sigma^{\Se}=\cov(\Sehatbm)$ and $\bm \Sigma^{\Sp}=\cov(\Sphatbm)$ as the sample covariance matrices of the similarity matrices divided by factors of $n_1$ and $n_0$, respectively.

A problem arises when a sample proportion of one (or zero) is observed. This is not unlikely to happen in the realistic scenario that either sensitivity or specificity of several models are close to one and the evaluation sample size is not large. In this case, say if $\Sehat_{m}=1$, the plug-in variance estimate $\Sehat_{m}(1-\Sehat_{m})/n_1$ collapses to zero. In effect, $\hat{\bm \Sigma}^{\Se}$ becomes singular and the statistical test procedure introduced in the next section is no longer directly applicable. Different approaches to deal with this problem have been described in the univariate context \cite{IBS, jung2017}. A popular approach is to employ the posterior mean of a Bayesian Beta-binomial model as a point estimator for the unknown mean. When taking a uniform, i.e. $\Beta(1,1)$, prior, this results in the point estimate $(u+1)/(n+2)$ where $u$ is the observed number of correct predictions, compare \citet[Chapter 9]{IBS}. Simply speaking, this amounts to adding two pseudo-observation - one correct and one wrong prediction. As a consequence, estimated proportions are shrunk (slightly) towards 0.5.

This idea can be transferred to the multivariate case ($S>1$) which is more complex as the correlation between empirical performances comes into play. \citet{SIMPle} recently derived a multivariate Beta-binomial model which can be employed for that matter. Besides marginal prior distributions, it also involves a prior on second-order moments, and thus the correlation structure. A vague, conservative prior can be specified as $S$ independent uniform distributions. Marginally, this again amounts to adding two pseudo observations (one correct and one false prediction) per model. Moreover, for each model pair half of the added correct pseudo-prediction is counted as common. We can then use the posterior mean and covariance as estimates $\Sehatbm$ and $\widehat{\bm \Sigma}^{\Se} = \widehat{\cov}(\Sehatbm)$ and independently apply the same approach for $\Sphatbm$ and $\widehat{\bm \Sigma}^{\Sp}$.
For finite sample sizes, the influence of this adaptation is conservative as mean estimates are shrunk towards $0.5$. Variance estimates are slightly inflated. In particular, all variances are guaranteed to be strictly greater than zero which solves the initially mentioned problem. Moreover, the correlation matrix is slightly shrunk towards the identity matrix. We have employed these regularized estimators in all numerical experiments in section \ref{sec:experiments}. Note that they are asymptotically equivalent to their naive counterparts and the asymptotic properties described in the next section thus remain the same. A detailed description can be found in appendix \ref{app:estimation}.

\subsection{Statistical inference}\label{sec:inference}

Our goal is to apply the so-called maxT-approach to hypothesis system \eqref{eq:hyp} \citep{hothorn2008}. This simultaneous test procedure has previously been utilized for the assessment of the overall classification accuracy of several binary classifiers \citep{EOMPM1, IMS}. This multiple test enables us to reduce the need for multiplicity adjustment in case several similar models are evaluated. Similarity is measured by the correlation of (empirical) performances which depends on the probability of a common correct prediction in our context \cite{EOMPM1}. The maxT-approach relies on determining a common critical value $c_\alpha$, which depends on the empirical correlation structure, such that each hypothesis $H_m$ is rejected if and only if 
\begin{align}
T_m = \min(\TSe_m, \TSp_m) > c_\alpha.
\end{align}
The test statistics are defined in a standard manner via 
\begin{align}\label{eq:teststat_def}
T_m^{\Se} = \frac{\Sehat_{m} - \Se_{0}}{\widehat{\se}(\Sehat_{m}) }
\end{align}
whereby $\widehat{\se}(\Sehat_m) = \sqrt{\Sehat_m(1-\Sehat_m)/n_1}$ is the estimated standard error of $\Sehat_m$. The test statistics $\TSp_m$ are defined accordingly. When only a single model is evaluated with regard to co-primary endpoints the tests regarding sensitivity and specificity each need to be conducted at local level $\alpha$. This is because the least favorable parameter configuration (LFC) is of the form
\begin{align}\label{eq:lfc_uni}
\Se=\Se_0 \wedge \Sp=1 \quad \text{or} \quad \Se = 1 \wedge \Sp = \Sp_0,
\end{align}
i.e. one parameter ($\Se$ or $\Sp$) is equal to one and the other parameter lies on the boundary of the null hypothesis. Similarly, for our case of $S$ candidate models, the two cases described in \eqref{eq:lfc_uni} are possible for each dimension $m \in \setS=\{1,\ldots,S\}$, resulting in $2^S$ potential LFCs. We could attempt to identify the single least favorable configuration in terms of multiple testing adjustment which depends on the resulting true correlation matrix. 
However, the computational burden of this extensive approach increases exponentially in $S$ as $2^S$ critical values would have to be computed.

We will instead focus on a more direct approach. Assume we knew if $\Delta_m^{\Se} = \Se_m - \Se_0 < \Sp_m - \Sp_0 =\Delta_m^{\Sp}$ or if $\Delta_m^{\Se} > \Delta_m^{\Sp}$ was true for each $m \in \setS$. For simplicity, we assume that $\Delta_m^{\Se}$ and $\Delta_m^{\Sp}$ are not exactly equal in the following. Define the indicator variables $b_m = \one(\Delta_m^{\Se} < \Delta_m^{\Sp})$, $\bb = (b_1,\ldots,b_S)$ 
and the corresponding diagonal matrix $\bm B = \diag(\bm b)$. 
Moreover, we define
\begin{align}\label{eq:B2T2}
{\BB_2} = \begin{pmatrix} \BB \\ \bm I_S - \BB  \end{pmatrix} \in \{0,1\}^{2S\times S} \quad \text{and} \quad \bm{T}_2 = \begin{pmatrix} \TSebm\\ \TSpbm \end{pmatrix} \in \mathbb{R}^{2S \times 1}
\end{align}
whereby $\bm I_S$ is the $S$-dimensional identity matrix. Finally, we can define the $S$-dimensional vector of test statistics
\begin{align}\label{eq:tstat_dagger}
\Tbm ^{\bm b} = ({\BB}_2)\tra \bm{T}_2 \in \mathbb{R}^S.
\end{align}
Our goal is to use this test statistic to obtain a multiple test for the hypotheses system  \eqref{eq:hyp} which can, due to the definition of the variables $b_m$, be rewritten as
\begin{align}\label{eq:hyp_dagger}
\setH_\setS = \{ H_m: b_m \Delta_m^{\Se} + (1-b_m) \Delta_m^{\Sp} \leq 0,\ m\in \setS \}.
\end{align}
The LFC for \eqref{eq:hyp_dagger} is specified by $\Sebm^{\bm b} = \bm b \Se_0 + (\bm 1_S - \bm b)$ and $\Spbm^{\bm b} = \bm b + (\bm 1_S - \bm b) \Sp_0$. That is to say, for each $m \in S$, $\Se_m$ is projected to $\Se_0$ and $\Sp_m$ is projected to one in the case $b_m=1$ and vice versa when $b_m=0$. Then, assuming $n_1/n \rightarrow \varrho \notin \{0,1\}$ as $n = n_1 + n_0 \rightarrow \infty$, under the LFC $\bm \theta^{\bm b} = (\Sebm^{\bm b}, \Spbm^{\bm b})$, we have
\begin{align}\label{eq:dist_result}
\Tbm^{\bm b}_{(n)} \stackrel{\D}{ \longrightarrow} \setN_S(\bm 0, \bm R^{\bm b}), \quad n \rightarrow \infty.
\end{align}
Hereby, the correlation matrix $\bm R^{\bm b}$ is defined as $\BB \bm R^{\Se} \BB + (\bm I_S - \BB) \bm R^{\Sp}  (\bm I_S - \BB) $. This result is proven in appendix \ref{app:inference}.

The vector $\bm b$ and thus the matrices $\BB$ and $\bm B_2$ are of course unknown in practice which makes the following two adjustments necessary. Firstly, we will replace $\bm T^{\bm b}$ with the vector $\bm T$ with entries ${T_m = \min(T_m^{\Se}, T_m^{\Sp}) \leq T^{\bm b}_m}$. Secondly, we need to estimate the correlation matrix $\bm R^{\bm b}$, in order to compute a critical value $c_\alpha$. An obvious candidate is
\begin{align}\label{eq:corr_est}
\hat{\bm R}^{\bm b} = \hat{\BB} \hat{\bm R}^{\Se} \BBhat + (\bm I_S - \BBhat)  \hat{\bm R}^{\Sp}  (\bm I_S - \BBhat)
\end{align}
with $\BBhat = \diag(\bbhat)$ and $\hat{\bb}_m=\one(\Sehat_m -\Se_0 < \Sphat_m -\Sp_0)$ and $\hat{\bm  R}^{\Se}$ and $\hat{\bm  R}^{\Sp}$ derived from the empirical covariances, compare section \ref{sec:estimation}. As all components in \eqref{eq:corr_est} are consistent estimators of their corresponding population quantities, $\hat{\bm R}^{\bm b}$ is a consistent estimator for $\bm R^{\bb}$ by virtue of the continuous mapping theorem, compare appendix \ref{app:inference}. 
To define a multiple test, we calculate a common critical value $c_\alpha$ such that 
\begin{align}\label{eq:fwer_control}
\pr(\max_{m \in \setS}(T^{\bm b}_{m}) \leq c_\alpha) \approx \Phi_{S}(\bm c_\alpha, \hat{\bm R}^{\bm b}) = \int_{-\infty}^{c_\alpha}\ldots \int_{-\infty}^{c_\alpha} \phi_{S}(\boldsymbol{x}, \hat{\bm R}^{\bm b}   )d\bm x = 1- \alpha,
\end{align}
under the 'estimated LFC' $\bm \theta^{\bbhat}$. Hereby $\phi_S$ and $\Phi_S$ denote the density and distribution function of the $S$-dimensional standard normal distribution with mean $\bm 0$ and covariance matrix $\hat{\bm R}^{\bm b}$. 
In practice, $c_\alpha$ can be found by numerical integration, e.g. with help of the \texttt{R} package \texttt{mvtnorm} \cite{mvtnorm}.

Altogether,  we have (a) asymptotic multivariate normality of $\bm T^{\bb}$, compare \eqref{eq:dist_result}, (b) a consistent estimate $\hat{\bm R}^{\bm b}$ of the correlation matrix via \eqref{eq:corr_est} and (c) $\bm T \leq \bm T^{\bm b}$ (deterministically). Following the argumentation of \citet{hothorn2008}, this allows to define a multiple test $\bm \varphi$ via $\varphi_m=1 \Leftrightarrow T_m > c_\alpha$ such that the FWER is asymptotically controlled at level $\alpha$.
We can also construct (e.g.) one-sided simultaneous confidence regions for $\Sebm$ and $\Spbm$ via
\begin{align}\label{eq:CI}
\text{CI}^{\Se}_{1-\alpha, m} =  \left(\Sehat_{m}- c_\alpha \cdot \widehat{\se}(\Sehat_{m}), \ 1\right) \quad \text{and} \quad 
\text{CI}^{\Sp}_{1-\alpha, m} =  \left(\Sphat_{m}- c_\alpha \cdot \widehat{\se}(\Sphat_{m}), \ 1\right).
\end{align}
Note, that the coverage probability $\pr_{\bm \theta}(\Sebm \in \CISebm_{1-\alpha} \wedge \Spbm \in \CISpbm_{1-\alpha})$ may be smaller than $1-\alpha$. Instead, due to the duality between confidence interval and test decision, we have asymptotically 
\begin{align}\label{eq:coverage}
\pr_{\bm \theta} \left(\bigcup_{m \in S} \left\{ \Se_m \notin \CISebm_{1-\alpha, m}\ \wedge\ \Sp_m \notin \CISpbm_{1-\alpha, m} \right\} \right) \leq \alpha.
\end{align} 
When setting $\alpha = 0.5$, the lower confidence bounds can be used as a corrected point estimators $\widetilde{\Sebm}$, $\widetilde{\Spbm}$ with the property that the overall probability that both $\widetilde{\Se}_m$ and $\widetilde{\Sp}_m$ overestimate their target for any $m \in \setS$ is asymptotically bounded by $50 \%$, compare \citet{EOMPM1}.


\section{Practical aspects}\label{sec:practical}

\subsection{Study planning}\label{sec:planning}

A prominent recommendation in machine learning and in medical testing is to strictly separate model development (training and selection) and model evaluation (performance assessment), compare section \ref{sec:introduction}. The most straightforward implementation of this strategy is to select a single final model prior to the evaluation study. Often times this is done by choosing the model with the highest empirical performance on a (hold-out) validation dataset out of the $M$ initially trained models, as described in section \ref{sec:assumptions}. 
When evaluating binary classifiers regarding their overall accuracy, previous work has shown that it is beneficial to evaluate multiple promising models simultaneously, e.g. all models within one standard error of the best validation model. \citet{EOMPM1, IMS} found that this so-called \textit{within\,1\,SE} selection rule leads to evaluation studies with higher power. Additionally, it allows to select the final model based on empirical test performances which in turn leads to an increased average final model performance. On the downside, the point estimate for the final model performance is upward biased which can however be corrected with help of the median-conservative estimator described at the end of the last section. To obtain multiplicity adjusted test decisions and corrected point estimates, \citet{EOMPM1, IMS} applied the so-called maxT-approach based on work by \citet{hothorn2008} which we adopted to the co-primary endpoint setting in section \ref{sec:inference} of this work.

To arrive at a suitable selection of prediction models for the evaluation study, more work is necessary when sensitivity and specificity are assessed simultaneously. We might still use the standard \textit{within\,1\,SE\,(Acc)} rule, but its success can be expected to depend on the disease prevalence $\varrho = \pr(Y=1)$. A simple adaptation is the  \textit{within\,1\,SE\,(bAcc)} rule based on the balanced accuracy
$\bAcc = (\Se+\Sp)/2$
instead of the overall accuracy $\Acc$. While the \textit{within\,1\,SE} approach is intuitive and has proven to work empirically, it lacks a throughout theoretical justification. In particular, the question which multiplier $k\geq 0$ results in the best \textit{within\,$k$\,SE} rule remains open.

In the following, we introduce a more elaborate algorithm for model selection which aims to derive the models to be evaluated in an optimal fashion. We are dealing with a subset selection problem, that is to say that our goal is to find the best subset $\setS \subset \setM$, whereby 'best' is yet to be defined. This topic has been extensively studied in the literature, see e.g. \citet[chapter\,9]{SDT2}, for a throughout treatment from the viewpoint of Bayesian decision theory. From this perspective, the usual approach is to pick the decision (here: the subset $\setS$) from an appropriate decision space (here: $\{0,1\}^M$) which minimizes the posterior expected loss, or equivalently maximizes the posterior expected utility. \citet[p.\,548]{SDT2}, give several examples for appropriate loss functions for the subset selection problem. For instance, we may define the utility of picking subset $\setS \subset \setM$ given the true parameter $\bm \theta=(\Sebm, \Spbm)$ and $\bm \vartheta = \min(\Sebm, \bm \Spbm+\Delta_0)$ (understood component-wise) as
\begin{align}\label{eq:other_utility} 
U_{\bm \theta}(\setS) = \max_{m \in S}\vartheta_m - \max_{m \in \setM}\vartheta_m - c|\setS|.
\end{align}
In \eqref{eq:other_utility}, we balance out the best selected model performance $\max_{m \in S}\vartheta_m$ versus the subset size $S=|\setS|$. Note that the second term in \eqref{eq:other_utility} is independent of $\setS$. This trade-off is necessary as we could easily maximize $\max_{m \in S}\vartheta_m$ by just selecting all models $\setS=\setM$. The parameter $c$ thus has to be set by the user to guide this trade-off. This is also the case for most other popular utility functions proposed in the literature \cite{SDT2}. This will be problematic in practice, because there is (again) no clear-cut solution how to specify the 'hyperparameter' $c$.

To overcome this issue, we will propose a utility function which has no hyperparameter in the above sense which needs to be specified. Hereby, it is important to realize that we are really dealing with a two-stage selection problem. First, our goal is to select a suitable subset $\setS \subset \setM$ of models before the evaluation study. Secondly, from $\setS$ we aim to select a final model $m^* \in \setS$ with help of the evaluation study for implementation in practice. 
As our ultimate goal is to obtain a model or medical test with high performance, i.e. high sensitivity and specificity, we propose to optimize the expected final model performance
\begin{align}\label{eq:efp_def}
\EFP(\setS) = \EFP_{\bm \theta}(\setS) = \E_{\p(\hat{\bm \theta} \given \bm \theta)}  \vartheta_{m^*} &=  \sum_{m\in \setS} \vartheta_{m} \pr_{\p(\hat{\bm \theta} \given \bm \theta)}(m^* = m) \\&= \sum_{m\in \setS} \min(\Se_{m}, \Sp_{m} + \Delta_0) \pr_{\p(\hat{\bm \theta} \given \bm \theta)}(m^* = m)
\end{align}
depending on the subset of models $\setS \subset \setM$ selected for evaluation.  The final model $m^*=m^*(\hat{\bm \theta})$ is chosen based on the empirical performances in the evaluation study and thus a random variable, unlike the fixed parameter $\bm \theta=(\Sebm, \Spbm)$. Our default final model choice will be $m^*= \argmax_{m \in \setS} \min(\TSe_m, \TSp_m)$, compare section \ref{sec:conclusion}. The expectation in \eqref{eq:efp_def} is taken with respect to $\p(\hat{\bm \theta} \given \bm \theta)$, the distribution of $\bm \thetahat$ in the evaluation study with $n$ observations which depends on $\bm \theta$. Assuming the prediction models are deterministic, $\p$ is determined by $\D^{n}_{(\bm X, Y)}$, the distribution of the feature-label data $(\bm X, Y)$ as described in the beginning of section \ref{sec:estimation}. Of course, both distributions are unknown in practice.

In formulation \eqref{eq:efp_def}, $\EFP$ depends on $\setS\subset\setM$ for which in principle there are $2^M$ choices. For the sake of simplicity we omit the choice $\setS=\emptyset$ in this work which corresponds to not conduct any evaluation study at all. In reality, this choice could of course be sensible, for instance if the empirical performances observed on the validation data suggest that none of the models satisfies our requirements $\theta_0=(\Se_0, \Sp_0)$. In order to simplify the decision problem which models to evaluate and its numerical optimization, we rank our models before choosing $\setS$ which leaves us with only $M$ choices instead of $2^M$, namely $\setS=\{1\}$, $\setS=\{1,2\}$, \ldots or $\setS=\setM$. Consequently, we may write $\EFP(S)$ instead of $\EFP(\setS)$ after ranking the models. 
How to rank models is not as obvious as in the single endpoint case concerned with the overall classification accuracy. Our default choice will be a ranking according to $\min(T^{\Se}_m(\setV), T^{\Sp}_m(\setV))$, whereby the test statistics are defined as in \eqref{eq:teststat_def} but based on the validation data. In other words, the models are ranked according to the validation evidence against the null hypotheses $H_m$.

The main idea behind our novel \textit{optimal\,EFP} selection rule is to maximize an estimate $\widehat{\EFP}(S)$ of $\EFP(S)$ before the evaluation study. Taking the viewpoint of Bayesian decision theory, we aim to maximize the posterior expected utility 
\begin{align}\label{eq:EFPhat}
\EFPhat(S,\ \pi({\bm \theta}\given \setV)) = \E_{\pi({\bm \theta}\given \setV)} \EFP_{\bm \theta}(S) =  \E_{\pi({\bm \theta}\given \setV)} \E_{p(\hat{\bm \theta} \given \bm \theta)}  \vartheta_{m^*}
\end{align}
at the time point of decision making, right before the evaluation study. The expectation is taken with regard to the posterior distribution $\pi=\pi({\bm \theta}\given \setV)$ of $\bm \theta$ given the validation data $\setV$.
More concretely, we assume that $\pi=(\pi^{\Se}, \pi^{\Sp})$ is composed of two multivariate Beta distributions, $\Sebm \sim \pi^{\Se}$ and $\bm \Spbm \sim \pi^{\Sp}$.
For that matter, we utilize a multivariate Beta (mBeta) distribution allowing arbitrary dependency structures between different Beta variables \cite{SIMPle}. As an initial prior distribution $\pi(\bm \theta)$ of $\bm \theta$, we use a vague (uniform) prior by default. A potential alternative would be to specify an informative prior distribution e.g. according to previous experimental data \citep{SIMPle}.

This approach allows us to sample $\bm \theta = (\Sebm, \Spbm)$ from $\pi$ and subsequently to sample estimates ${\bm \thetahat = (\Sehatbm, \Sphatbm)}$ from $\p=p(\bm \thetahat \given \bm \theta)$. In our case, the sampling distribution $\p$ of $\bm \thetahat$ is a multivariate Binomial distribution with mean vector $\bm \theta$ and correlation structure $\bm C_{\bm \thetahat}$. The correlation structure is also sampled from $\pi$ in terms of corresponding mixed second-order moments. Utilizing this Bayesian model allows us to forecast both, true parameter values and their respective estimates including the dependency structure of both. This enables a simulation of the selection process in the evaluation study and the inherent selection-induced bias problem. In our numerical implementation, we iteratively repeat this process until our computational budget has been exhausted and finally average the results to approximate the (double) expectation in \eqref{eq:EFPhat} for all $S$. Hereby, we specify a maximal number of models $S_{max}$ for evaluation in advance to reduce the computational burden of the optimization. A solution to the subset selection problem in dependence of the validation data $\setV$ is then simply $\setS = \{1,\ldots,S^*\}$ with
\begin{align}\label{eq:Sopt}
S^*=S^*(\setV) = \argmax_{1\leq S \leq S_{max} } \widehat{\EFP}(S, \pi({\bm \theta}\given \setV)). 
\end{align}
Given the assumed prior distribution, the initial model ranking, the statistical modeling assumptions and the numerical approximations, the procedure results in an approximate Bayes action as we approximately maximize the posterior expected utility $\EFP(S)$. The numerical implementation of the \textit{optimal\,EFP} selection rule is described in more detail in appendix \ref{app:optimal}.

Figure \ref{fig:algo_example} illustrates the trade-of between selecting too few and too many models for evaluation. The curves represent $\widehat{\EFP}(S)$ simulated by the above sketched algorithm for validation performance estimates $\bm \thetahat(\setV)$
based on $n_\setV=100$ observations taken from a single exemplary instance from our simulation database, compare section \ref{sec:mle_sim}.
The three different curves correspond to different evaluation sample sizes $n_\setE \in \{100, 200, 400\}$, compare section \ref{sec:mle_sim}. While too few models make it easy to miss a truly good candidate (because the validation ranking might be incorrect), too many models make it hard to identify a truly adequate solution. When more test observations are available, a higher number of models can be evaluated before the test set 'overfitting' occurs and $\EFPhat(S)$ therefore decreases slower in $S$. For the numerical experiments presented in this work (section \ref{sec:mle_sim}), we do not use \eqref{eq:Sopt} to determine $S^*$, as one would naturally do, but rather select the smallest $S$ with comparable $\widehat{\EFP}(S)$. By default, we define $S^*$ as the smallest $S$ such that $\widehat{\EFP}(S^*)$ is still within one standard error of $\max_{1 \leq S \leq S_{max}}\widehat{\EFP}(S)$. This is indicated in figure \ref{fig:algo_example}. For $n_\setE=200$ the maximizer of $\widehat{\EFP}(S)$ is $S=30$ which yields $\EFPhat(30)=82.51\%$. The simulation standard error associated to this estimate is $0.17\%$. If we select the smallest $S$ with comparable $\EFPhat(S)$ we obtain $S^*=15$ which gives $\EFPhat(15)=82.49\%$. Thus, the expected final performance is almost identical but choosing fewer models will naturally result in a smaller adjustment for multiplicity and estimation bias in the evaluation study. 

\begin{figure}
	\centering
	\includegraphics[width=0.8\linewidth]{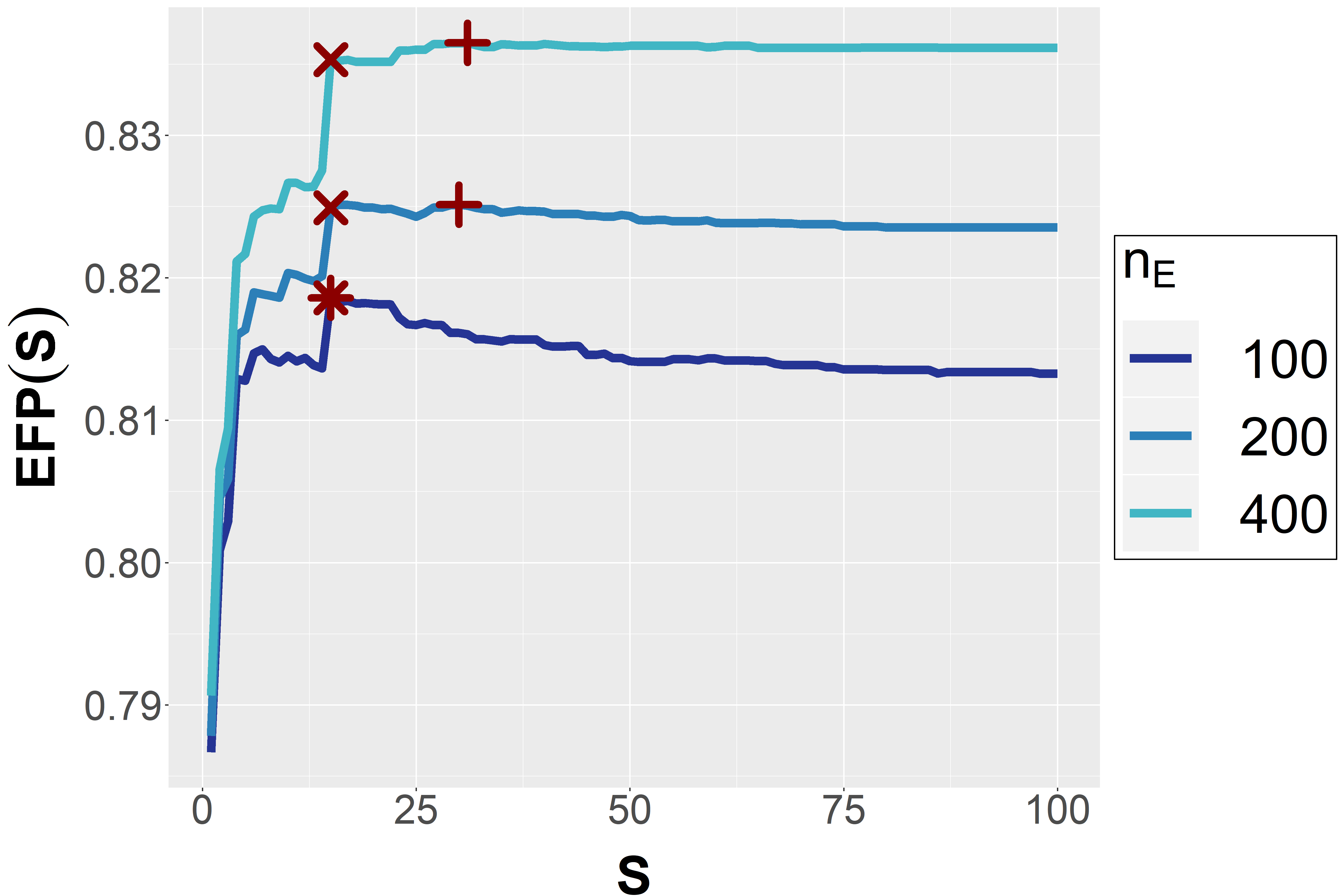}
	\caption{	
		Optimization of the expected final model performance $\EFP(S)$ as a function of the number of models to be evaluated. The figure shows $\EFPhat(S)$ after $500$ simulation runs for different evaluation sample sizes $n_\setE = 100, 200, 400$ for a single exemplary dataset from our simulation database. Besides the maximum ($\bm +$), the smallest $S$ with comparable $\EFPhat(S)$ (within one standard error of the maximum) is marked ($\bm \times$). For $n_\setE=100$, both points ($\bm +$, $\bm \times$) coincide.
	}
	\label{fig:algo_example}
\end{figure}

\subsection{Study conclusion}\label{sec:conclusion}

Our framework provides a multiple test decision for the hypotheses system \eqref{eq:hyp}. After the evaluation study, one might therefore partition the set of candidate models $\setM=\setP \cup \setN$ into positively and negatively evaluated models $\setP=\{m \in \setS\subset \setM: \varphi_m=1 \}$ and $\setN = \setS \setminus \setP$. In the case $\setP= \emptyset$, no performance claim can be made and the evaluation study has failed in a strict confirmatory sense. When $|\setP|=1$, only a single model meets the requirements and the situation is clear. However, if $|\setP|>1$, the investigator has the choice which of the models $m \in \setP$ should be implemented in practice.

When applying the maxT-approach, $m^*=\argmax_{m \in \setS} T_m$ is a natural final model choice. In the co-primary endpoint setting, this corresponds to $m^* = \argmax_{m \in \setS} \min(T^{\Se}_m, T^{\Sp}_m )$, i.e. we select the model with the highest evidence that both individual hypotheses are false. %
It might however be reasonable to select the final model according to other criteria. 
For instance, one may declare the final model to be 
\begin{align}\label{maxW}
m^{*}=
\argmax_{m \in \setP} (w \cdot \Sehat_m + (1-w)\cdot \Sphat_m), \quad w \in (0,1).
\end{align}
This means we choose the final model such that a weighted average of empirical sensitivity and specificity is maximized, given the null hypothesis was rejected. When $w=0.5$ this corresponds to the maximization of the (empirical) balanced accuracy ($\bAcc$), or equivalently, the well known Youden-Index from all models $m\in\setP\subset \setS$ \cite{youden1950}. Further criteria which are not directly tied to the discriminatory performance such as model interpretability or calibration can of course also be considered.


\section{Numerical experiments}\label{sec:experiments}

\subsection{Machine learning and evaluation}\label{sec:mle_sim}

\subsubsection{Setup}\label{sec:mle_sim_setup}

The goal of this first simulation study is to assess the properties of different model selection rules and our statistical inference approach under realistic parameter settings in the machine learning context. To this end, we investigate different methods on a large simulation database which was generated for related methodological comparisons \cite{IMS}. While the overall binary classification accuracy was our only performance measure of interest in earlier work, we turn to the investigation of sensitivity and specificity as co-primary endpoints in the following. The simulation database consists of $144,000$ instances of the complete machine learning pipeline. A single instance consists of training ($\setT$), validation ($\setV$) and evaluation ($\setE$) datasets which are all comprised of feature-label observations $(\bm x, y)$ sampled from the same distribution $\D_{(\bm X, Y)}$.

Different binary classification models were trained on the learning data $\setL=\setT \cup \setV$, compare section \ref{sec:assumptions}.
For that matter, we employed four popular learning algorithms (elastic net, classification trees, support vector machines, extreme gradient boosting) with help of the \texttt{R} package \texttt{caret} \cite{APM, caret}. More concretely, we trained $M=200$ initial candidate models, 50 models obtained from randomly sampled hyperparameters per algorithm, on the training data $\setT \subset \setL$. The distinct validation data $\setV = \setL \setminus \setT$ is then used for model selection. The models $\setS \subset \setM= \{1,\ldots, M\}$ selected for the evaluation study are then re-trained on the complete learning data $\setL= \setT \cup \setV$ as this can be expected to slightly increase their predictive performance. The selected models $\setS \subset \setM$ then undergo a final assessment on the independent evaluation data $\setE$ which is supposed to mimic a diagnostic accuracy study. The goal of this final evaluation study is to estimate sensitivities and specificities of the selected models and to obtain a test decision regarding hypotheses system \eqref{eq:hyp_ext}. Ultimately, a final model $m^* \in \setS$ is selected to be implemented in practice, if positively evaluated, compare section \ref{sec:conclusion}.

For each individual simulation instance, the distinct datasets $\setL$ and $\setE$ are sampled from the same joint probability distributions $\D_{(\bm X, Y)}$. Different distributions $\D_{(\bm X, Y)}$ with varying characteristics have been specified to generate the entire simulation database. This covers linear and non-linear tasks (ratio 1:2), independent and dependent feature distributions (ratio 1:1) and different disease prevalences $\varrho\in\{ 0.15, 0.3, 0.5\}$ (ratio 1:1:1). The learning data $\setL$ is of size $n_\setL \in \{400, 800\}$ (ratio 1:1). The model selection was always conducted on a hold-out validation dataset $\setV \subset \setL$ of size $n_\setV = n_\setL/4$. 
Each learning dataset (including the trained models) was used twice, once in connection with each of the considered evaluation dataset sizes $n_\setE \in \{400, 800\}$ (ratio 1:1).

The true model performances $\theta =(\Se, \Sp)$ of all models are approximated with high precision on a large population dataset ($n_\setP=100,000$). This dataset was not used for any other purposes. The truly best model is defined in line with the study goal outlined in section \ref{sec: goals} as
\begin{align}\label{eq:trulybest}
m^\op = \argmax_{m \in \setM} \vartheta_m = \argmax_{m \in \setM} \min(\Se_m, \Sp_m + \Delta_0).
\end{align}
The parameter $\Delta_0$ is set to $0$ for all experiments which expresses equal importance of sensitivity and specificity. The corresponding maximal performance is denoted as $\vartheta_\op = \vartheta_{m^\op}$. 
The final model $m^*\in \setS$ is chosen based on the evaluation data as $m^*=\argmax_{m\in\setS} T_m$, compare section \ref{sec:conclusion}. 
In case of a tie, $m^*$ is chosen randomly from the set $\argmax_{m\in\setS} T_m$. 
For a throughout description of the simulation database we refer to \citet{IMS}. \texttt{R} packages and scripts related to the simulation study are publicly accessible.\footnote{\url{https://maxwestphal.github.io/SEPM.PUB/} (accessed March 22, 2020)} 

We initially divided the simulation database into two parts via a stratified random split with regard to the above mentioned factors. In effect, all mentioned ratios are identical in both parts. The smaller subset with $24,000$ instances was used in advance to perform a few methodological comparisons with the goal to fine-tune certain aspects of the \textit{optimal\,EFP} selection rule described in section \ref{sec:planning}. For instance, we altered the way models are ranked initially and several details associated to numerical complexity of the algorithm (number of iterations, convergence criterion), compare appendix \ref{app:optimal}. The final algorithm was then fixed and investigated on the main part which consists of $120,000$ instances, as described in the following.

\subsubsection{Goals}\label{sec:mle_sim_goal}

Our overall goal is an assessment of important operating characteristics of the employed evaluation strategies. An evaluation strategy is comprised of a selection rule (before evaluation) in and a statistical testing procedure (for evaluation). The latter will always be conducted via the maxT-approach presented in section \ref{sec:inference}. When only a single model is selected, the maxT-approach reduces to the standard procedure for co-primary endpoints, namely two independent (unadjusted) $Z$-tests for sensitivity and specificity which both need to reject in order claim superiority.

The primary objective of this simulation study is therefore the comparison of different selection rules. Hereby, we are mainly interested in the performance of our novel Bayesian approach designed to optimize the expected final model performance (EFP), compare section \ref{sec:planning}. This approach is referred to as \textit{optimal\,EFP} selection rule. Its main competitor is the \textit{within\,1\,SE} rule for which $\setS$ is defined as all models $m\in \setM$ with preliminary balanced accuracy $\widehat{\bAcc}_m(\setV) = (\Sehat(\setV)+\Sphat(\setV))/2$ (estimated on $\setV$) not more than one standard error below the maximal $\widehat{\bAcc}_m(\setV)$. A similar rule has lead to vast improvements relative to the traditional \textit{default} approach in previous work concerned with the overall accuracy \cite{EOMPM1, IMS}. The \textit{default} rule selects the single best model in terms of $\widehat{\bAcc}(\setV)$ for evaluation and will serve as a secondary comparator. Additionally, we consider the \textit{oracle} selection rule which cannot be implemented in practice but gives an insight on the theoretically achievable performance (of selection rules). The \textit{oracle} rule always selects the truly best model $\setS= \{ m^\op\}$ as defined in \eqref{eq:trulybest}. To reduce the computational complexity of the simulation study we thresholded the number of evaluated models to $S_{max}=\lfloor\sqrt{n_\setE}\rceil$. Performance ties are resolved by including all tied models for evaluation. In effect, the \textit{default} and \textit{oracle} rule also occasionally select more than one model for evaluation.

Our main research hypothesis going into the simulation is that the \textit{optimal\,EFP} improves the expected final performance (EFP) and statistical power compared to the \textit{within\,1\,SE} rule while having comparable estimation bias. Moreover, we require that the FWER is still controlled at the desired level which is set to $\alpha=2.5\%$ (one-sided) in all scenarios.

\subsubsection{Results}\label{sec:mle_sim_results}

\begin{figure}[ht!]
	\captionsetup[subfigure]{justification=centering}
	\begin{subfigure}[c]{0.5\textwidth}
		\includegraphics[width=\textwidth]{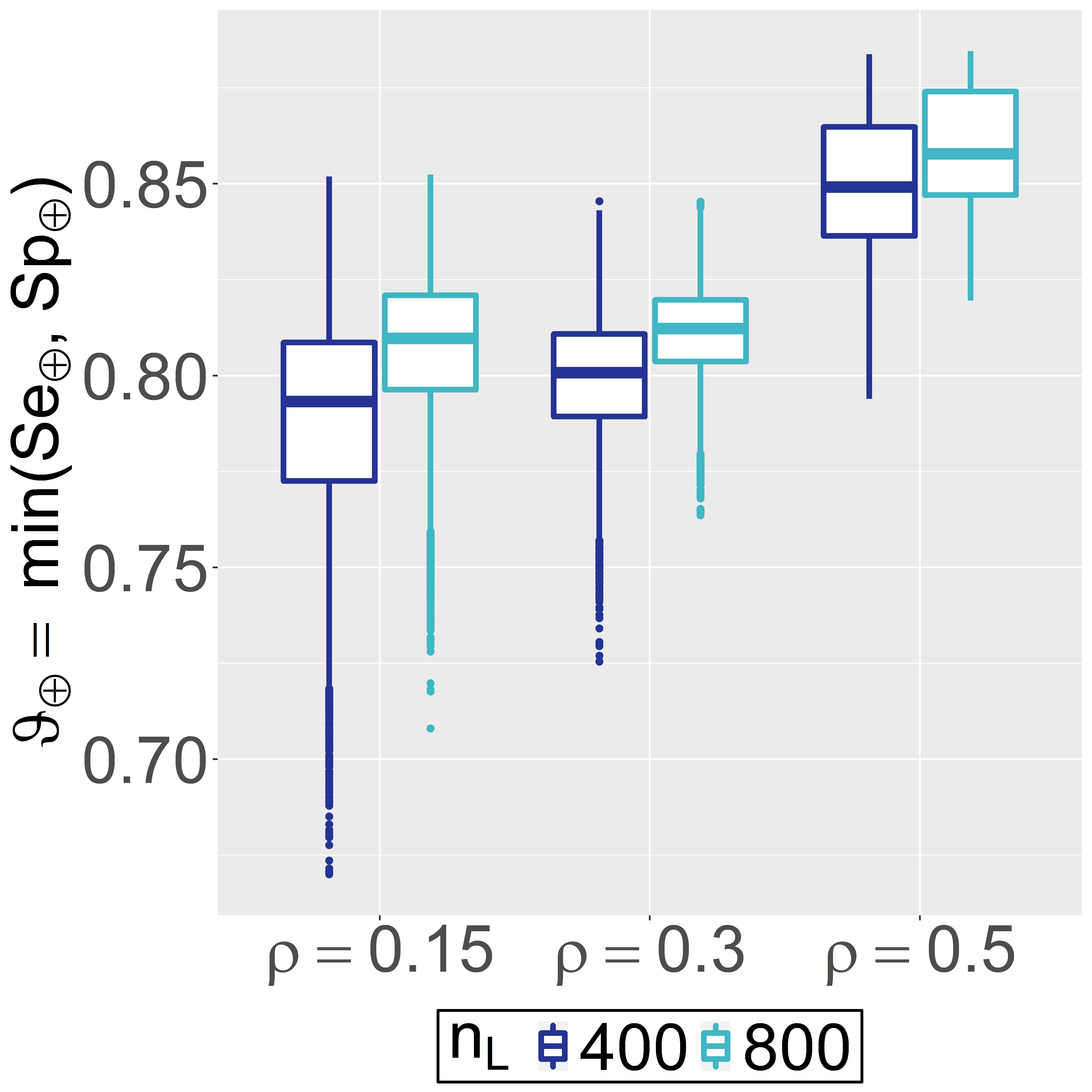}
	\end{subfigure}
	\begin{subfigure}[c]{0.5\textwidth}
		\includegraphics[width=\textwidth]{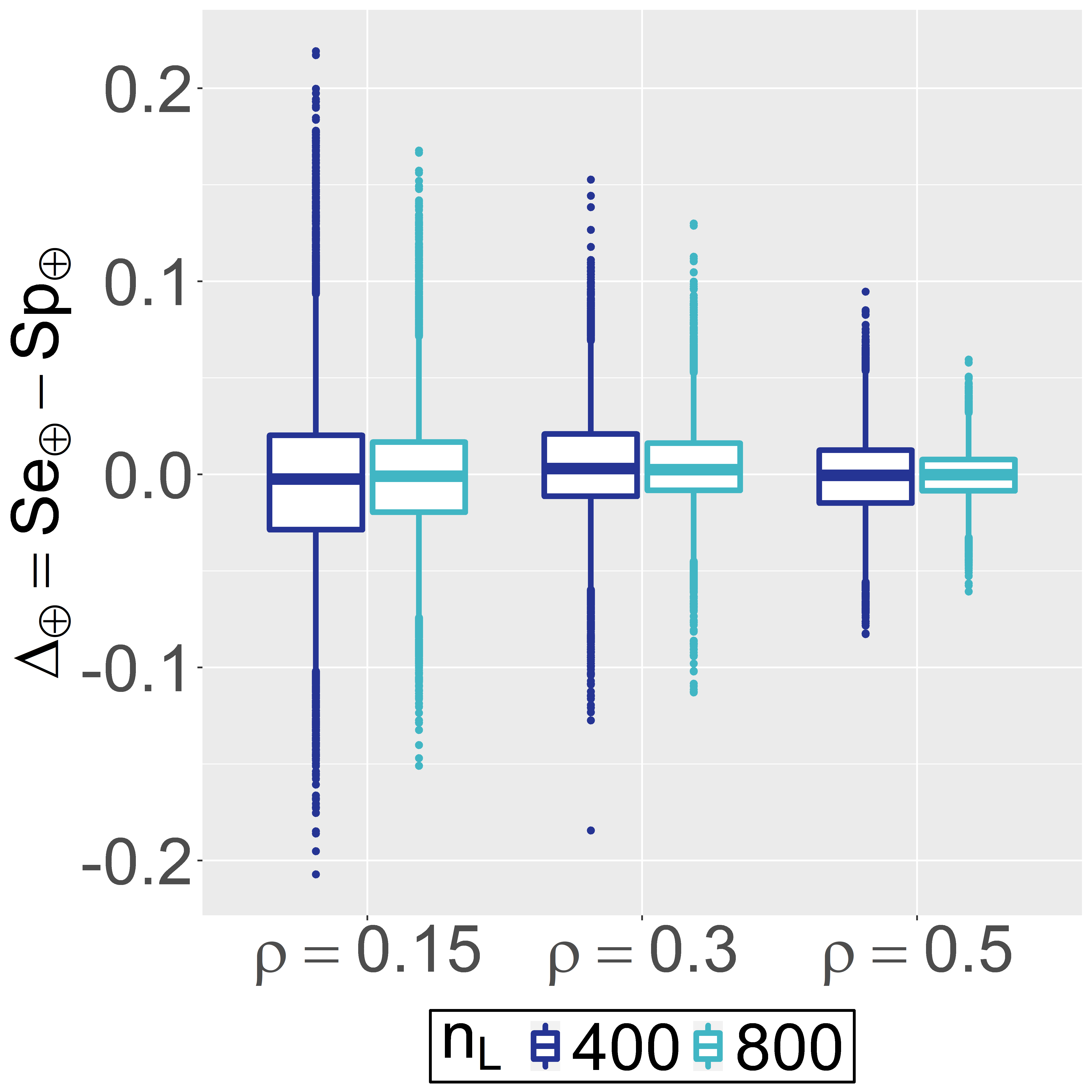}
	\end{subfigure}
	\caption{	
		Distribution of  $\vartheta_{\op}$ and $\Delta_{\op}$ stratified for disease prevalence $\varrho$ and learning sample size $n_\setL$. Each boxplot is based on 10,000 distinct points which amounts to 60,000 unique simulation instances in total.
	}
	\label{fig:mle_sim_eda}
\end{figure}

Firstly, we investigate how well the feature-label relationships can be learned by the considered algorithms. 
Figure \ref{fig:mle_sim_eda} (left) shows the distribution of the truly best model performance $\vartheta_{\op}$ over all unique $60,000$ learning instances. The results are stratified for disease prevalence $\varrho$ and learning sample size $n_\setL$. We observe that $\vartheta_{\op}$ is increasing in both these factors, as expected. Overall, the median of $\vartheta_{\op}$ lies between $79\%$ and $86\%$ which we consider to be a very realistic range. Moreover, figure \ref{fig:mle_sim_eda} (right) shows the distribution of $\Delta_{\op} = \Se_{\op}- \Sp_{\op}$ which indicates how balanced the accuracy of the best model is across the diseased and healthy subpopulations. Note that for the \textit{oracle} rule, the model selection is $\setS = \{m^\op\}$ and thus $m^* = m^\op$ is also the final model.

\begin{figure}[t!]
	\includegraphics[width=\linewidth]{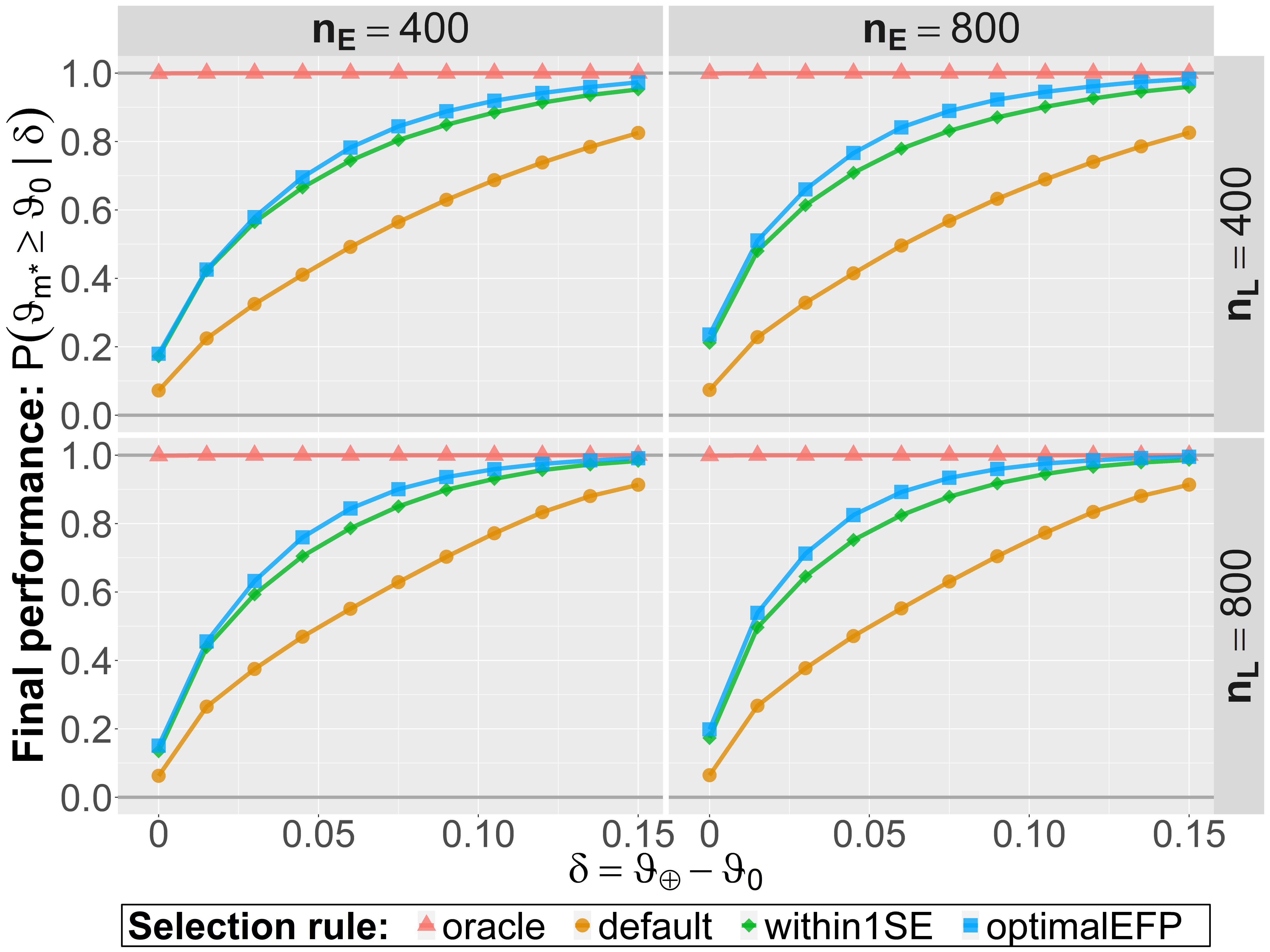}
	\caption{Probability of the event that the final model performance $\vartheta_{m^*}=\min(\Se_{m^*}, \Sp_{m^*}+\Delta_0)$ exceeds a given threshold $\vartheta_0 = \vartheta_{\op} - \delta$. $\vartheta_{\op}$ is the true performance of the best candidate model. Each scenario of learning and evaluation sample sizes $(n_\setL, n_\setE)$ is based on $30000$ repetitions of the complete machine learning pipeline. 
	}
	\label{fig:mle_sim_fp}
\end{figure}

Figure \ref{fig:mle_sim_fp} illustrates the true performance of the final selected model $m^* \in \setS$ whereby $\setS$ depends on the employed selection rule.  
Depending on the threshold value $\vartheta_{0} = \Se_{0} = \Sp_{0}$, figure \ref{fig:mle_sim_fp} shows the (simulated) probability ${\pr(\vartheta_{m^*} \geq \vartheta_0\given \delta)}$ that the final chosen model has at least a performance of $\vartheta_0 = \vartheta_\op - \delta$, $\delta>0$. The general picture is that the \textit{optimal\,EFP} and \textit{within\,1\,SE} rules both clearly outperform the \textit{default} rule. 
Concerning our main comparison, we observe that the \textit{optimal\,EFP} approach uniformly outperforms the \textit{within\,1\,SE} rule. The effect regarding ${\pr(\vartheta_{m^*} \geq \vartheta_0\given \delta)}$ is rather small, e.g. $0.791-0.735 =0.056$ (99\% CI: 0.053-0.059) at $\delta=0.05$ when averaging over all scenarios ($n_\setL,n_\setE$).
We also compared the mean expected performance $\E\vartheta_{m^*}$ of both rules, which are $0.788$ and $0.782$, respectively. This comparison thus also shows a slight advantage of the \textit{optimal\,EFP} rule of an average performance gain of $0.007$ (99\% CI: 0.006-0.007). Further details are provided in table \ref{tab:mle_sim_table} in appendix \ref{app:experiments}.

\begin{figure}[t!]
	\includegraphics[width=\linewidth]{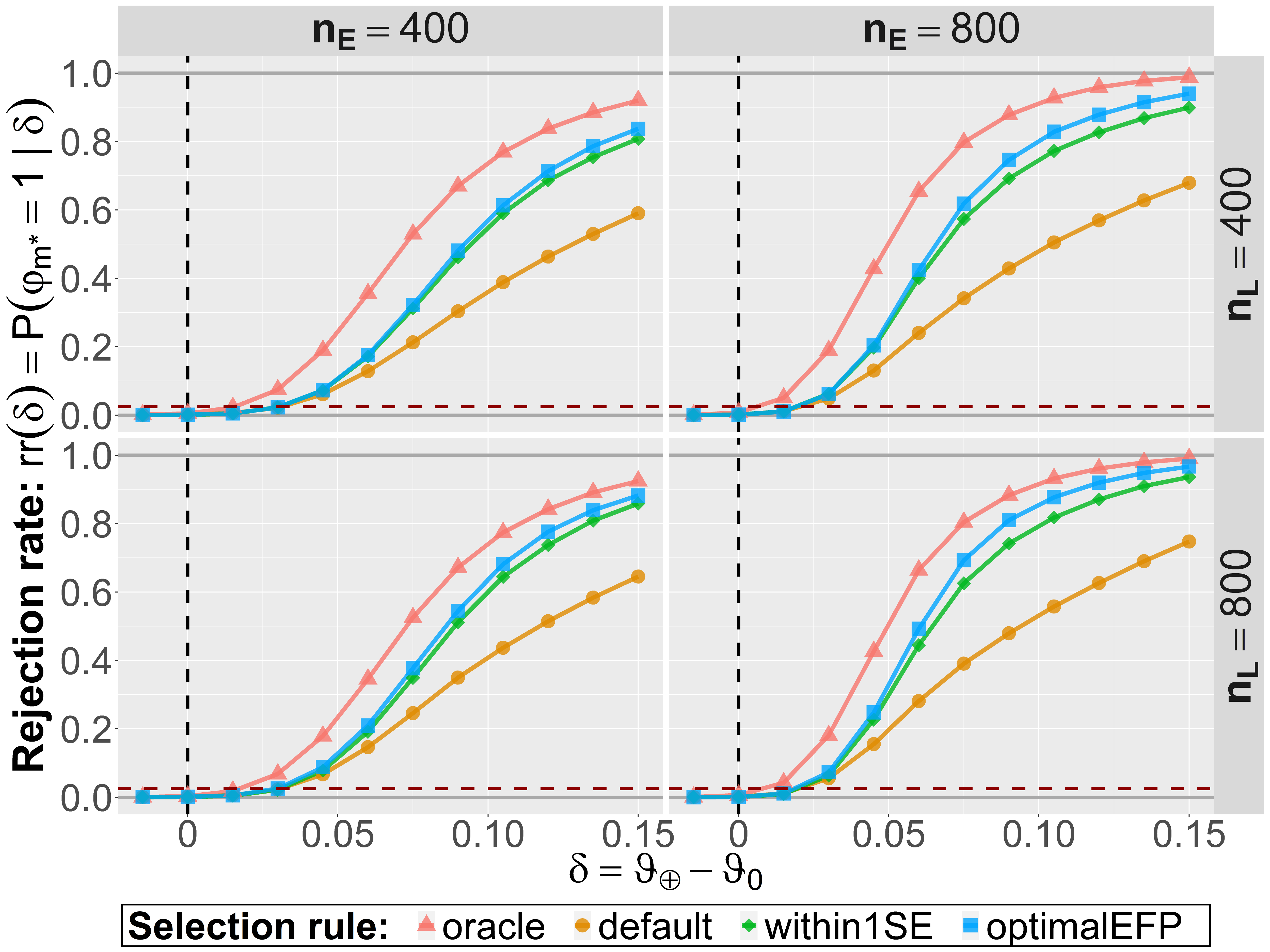}
	\caption{Rejection rate for the global null hypothesis (all prediction models have performance $\vartheta_m \leq \vartheta_0$ ) in the evaluation study after models have been selected for evaluation by specified selection rule. $\delta=\vartheta_{\op}-\vartheta_0$ specifies if the global null is true ($\delta \leq 0$) or false ($\delta > 0$). Each scenario of learning and evaluation sample sizes $(n_\setL, n_\setE)$ is based on $30000$ repetitions of the complete machine learning pipeline. The significance level is $\alpha=2.5\%$ (one-sided, dashed horizontal line).
	}
	\label{fig:mle_sim_rr}
\end{figure}

Figure \ref{fig:mle_sim_rr} shows the rejection rate $rr(\delta) = \pr(\varphi_{m^*}=1 \given \delta)$ as a function of $\delta=\vartheta_\op - \vartheta_0$. 
The case $\delta=0$ or $\vartheta_0 = \vartheta_\op$ corresponds to the smallest threshold value $\vartheta_0$ such that the global null hypothesis (none of the initially trained models has a high enough performance) is true, compare \eqref{eq:hyp_ext}. Thus, for $\delta \leq 0$ the global null is true and the rejection rate $rr(\delta)$ should be bounded by $\alpha=2.5\%$ for all selection rules. Note that this parametrization via $\delta$ is necessary as $\vartheta_\op$ is not fixed over all simulation instances as illustrated in figure \ref{fig:mle_sim_eda}. The results clearly show that FWER control is given for all rules in all scenarios. In appendix \ref{app:experiments}, table \ref{tab:mle_sim_table}, we also show the FWER when the (pre-)selection process is not incorporated in the test decision. More precisely, we are then looking at the type I error rate of the final statistical test for hypothesis system \eqref{eq:hyp} which is equal to $\pr(\varphi_{m^*}=1 \given \vartheta_0 = \vartheta_{m^*})$. We refer to this quantity as conditional FWER which is necessarily larger than $rr(0)$ but still controlled, i.e. smaller than $\alpha$ in all cases. Note that, strictly speaking, it is only an upper bound of the proper conditional FWER $\pr(\varphi_{m^*}=1 \given \vartheta_0 = \max_{m\in \setS}\vartheta_m)$ which was not 'recorded' in our numerical experiments.

In the situation $\delta>0$, the threshold value $\vartheta_0$ is small enough such that the global null is no longer true, i.e. there is at least one model $m \in \setM$ with true performance $\vartheta_m > \vartheta_0$ which we seek to detect. Larger rejection rates are thus positive as they correspond to statistical power. Concerning power, the ranking of selection rules is the same as in the expected final performance comparison (figure \ref{fig:mle_sim_fp}). The multiple testing approaches \textit{optimal\,EFP} and \textit{within\,1\,SE} again cleary outperform the \textit{default} rule. The power increase is up to $30\%$, depending on the scenario and $\delta$. The power difference between \textit{optimal\,EFP} and \textit{within\,1\,SE} rule is always positive but only noteworthy in the case $n_\setE=800$.

Table \ref{tab:mle_sim_table} in appendix \ref{app:experiments} shows further comparisons of the investigated selection rules. For instance, the properties of the corrected point estimators $\tilde{\theta}_{m^*} = (\widetilde{\Se}_{m^*}, \widetilde{\Sp}_{m^*})$ are assessed, compare section \ref{sec:inference}. When only a single model is evaluated, $\tilde{\theta}$ corresponds to $\thetahat$, the raw (slightly regularized) point estimate, compare section \ref{sec:estimation}. For example, the combined mean absolute deviation 
\begin{align}
\MAE_2 = \E (|\widetilde{\Se}_{m^*}-\Se_{m^*} |+ |\widetilde{\Sp}_{m^*}-\Sp_{m^*}|)/2
\end{align}
is decreased when multiple models are evaluated simultaneously (\textit{optimal\,EFP}: 0.042; \textit{within\,1\,SE}: 0.049, averaged over both cases $n_\setE \in \{400, 800\}$) compared to the \textit{default rule} rule (0.074). This might seem counterintuitive at first, but is in line with previous findings \citep{IMS}. This observation can be explained with the on average increased final model performance $\vartheta_{m^*}=\min(\Se_{m^*}, \Sp_{m^*})$ which results in a decreased variance of the binomially distributed estimators. In contrast, the estimation bias is more pronounced when multiple models are evaluated. Recall that we have $\pr(\widetilde{\vartheta}_{m^*} > \vartheta_{m^*})\leq 0.5$ asymptotically by construction as we use the corrected estimator $\widetilde{\vartheta}_{m^*} = \min(\widetilde{\Se}_{m^*}, \widetilde{\Sp}_{m^*})$, compare section \ref{sec:inference}. Finally, we note that the \textit{optimal\,EFP} rule on average selects roughly two to three fewer models for evaluation compared to the \textit{within\,1\,SE} rule, compare table \ref{tab:mle_sim_table} in appendix \ref{app:experiments}.

\subsection{Least favorable parameter configurations}\label{sec:lfc_sim}

In the previous section, the operating characteristics of our simultaneous inference framework were assessed in a variety of realistic scenarios in the predictive modeling context. In contrast, we investigate the worst case type I error rate of the maxT-approach for different sample sizes in the following. For that matter, we simulate synthetic data under least favorable parameter configurations (LFCs) as described in section \ref{sec:inference} and also parameters close to the LFC.

We assume that $S \in \{1, 10, 20\}$ models are assessed on the evaluation data and consider the case when $\Se_0$ and $\Sp_0$ are identical. For each simulation run, we randomly draw a binary vector $\bb$ such that $|\bb|=S/2$. When $S=1$, we set $\bb = (1)$. The true parameters $\Sebm$ and $\Spbm$ are defined as
\begin{align}
\Se^{\bm b, \epsilon}_m &= b_m  (\Se_0 - (m-1) \epsilon) + (1 - b_m), \\
\Sp^{\bm b, \epsilon}_m &= b_m + (1 - b_m) (\Sp_0 - (S-m) \epsilon).
\end{align}
Hereby, the parameter $\epsilon \in \{0, 0.001, 0.002\}$ expresses how close we are to a true LFC which corresponds to $\epsilon =0$. For instance, for $S=10$, $\Se_0=\Sp_0=0.8$, $\epsilon=0.001$, a possible (nearly least favorable) parameter configuration is given by
\begin{alignat}{12}\label{eq:lfc_example}
\bm b                     &=\ &&(1     &&, 1     &&, 0     &&, 1     &&, 0     &&, 0     &&, 0     &&, 1     &&, 1     &&, 0    &&), \\
\Sebm^{\bm b, \epsilon}_m &=\ &&(0.800 &&, 0.799 &&, 1     &&, 0.797 &&, 1     &&, 1     &&, 1     &&, 0.793 &&, 0.792 &&, 1    &&), \\
\Spbm^{\bm b, \epsilon}_m &=\ &&(1     &&, 1     &&, 0.793 &&, 1     &&, 0.795 &&, 0.796 &&, 0.797 &&, 1     &&, 1     &&, 0.800&&).
\end{alignat}
If $\epsilon$ were $0$ instead of $0.001$ in \eqref{eq:lfc_example}, all elements in $\Sebm^{\bm b, \epsilon}$ and $\Spbm^{\bm b, \epsilon}$ different from $1$ would change to ${\Se_0=\Sp_0=0.8}$ (assuming the same $\bb$). We mainly consider a (fixed) disease prevalence of $\varrho=0.2$. Hence, there are $n_1 = \varrho n$ diseased and $n_0 = n-n_1$ subjects, whereby $n \in \{200, 400, 800, 4000, 20000\}$ in our simulation. Additionally, the scenarios with $\epsilon=0$ were also repeated with a balanced class distribution ($\varrho=0.5$). In all simulations, the true correlation between all estimators which belong to non-trivial variables (true mean not equal to one), e.g. $\Sehat_1$ and $\Sehat_2$ in the example \eqref{eq:lfc_example}, is set to 0.5. Results from a sensitivity analysis regarding this specification are reported at the end of this section. Any empirical sensitivity is of course independent from any empirical specificity. Data generation, i.e. sampling of the similarity matrices $\bm Q^{\Se}$ and $\bm Q^{\Sp}$ (section \ref{sec:estimation}) based on the above specifications, was conducted with help of the \texttt{bindata} package in \texttt{R} \cite{leisch1998, bindata}. The number of simulation runs was set to $N_{sim} = 10,000$ which results in a (point-wise) upper bound for the standard error of the simulated FWER of $0.5/\sqrt{10.000}=0.005$.

\begin{figure}[t!]
	\includegraphics[width=\linewidth]{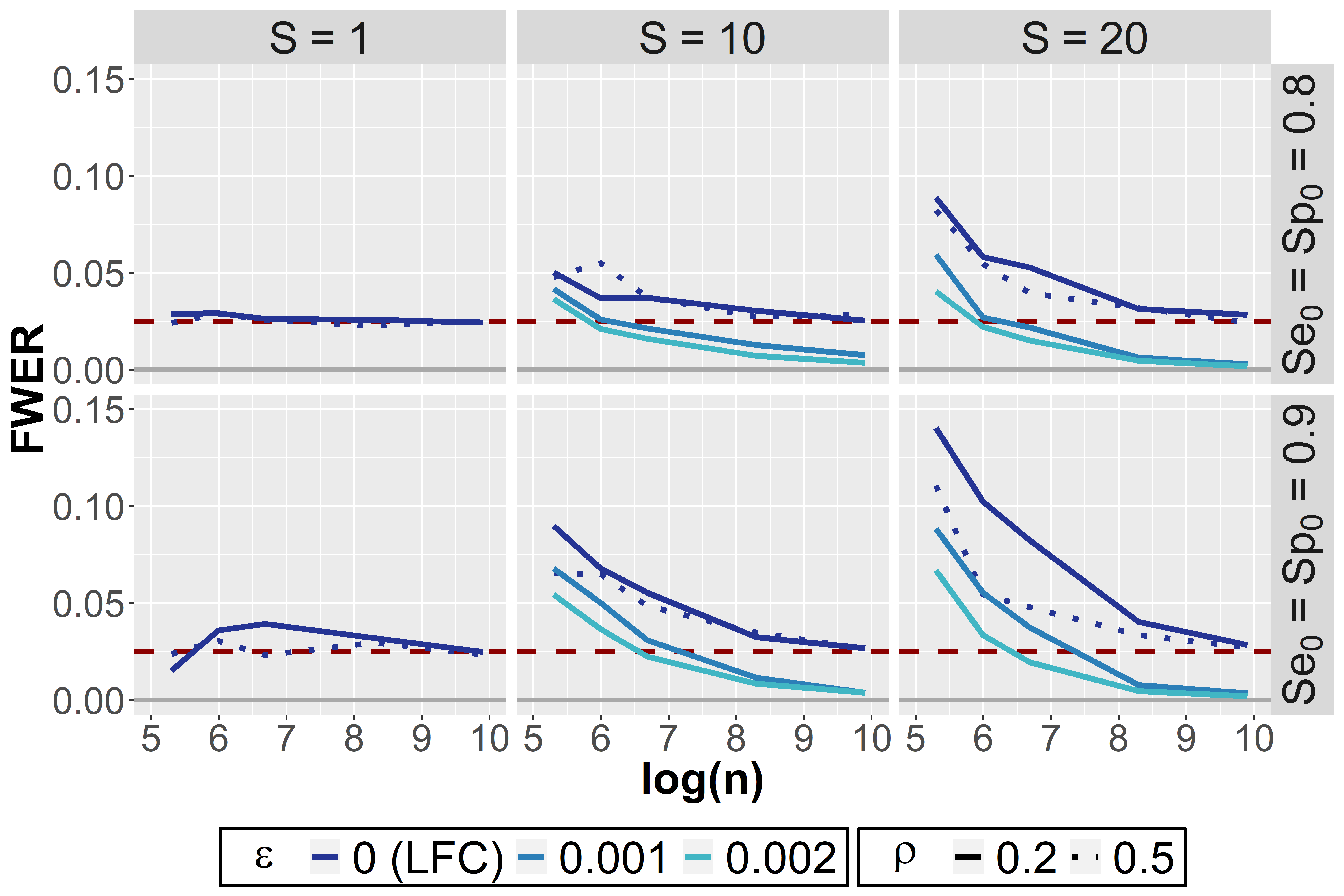}
	\caption{Family-wise error rate (FWER) in dependence of the total sample size $n=n_1+n_0$ (logarithmic scale) under different parameter configuration (specified by $\vartheta_0, \epsilon$) and disease prevalences $\varrho=n_1/n$. In all cases, $S/2$ (randomly selected) models have $\Se_m=\Se_0$ and $\Sp_m=1$ and vice versa for the remaining $S/2$ models. For $S=1$, $\Se = \Se_0$ and $\Sp=1$. The target significance level is $\alpha=2.5\%$ (one-sided). $N_{sim} = 10,000$. 
	}
	\label{fig:lfc_sim}
\end{figure}

Figure \ref{fig:lfc_sim} shows the simulation results stratified by all relevant factors. In all scenarios, the results match the theoretical findings from section \ref{sec:inference} in that the asymptotic FWER under LFCs converges to $\alpha=0.025$ as a function of the sample size $n$.
When the disease prevalence is increased from $\varrho=0.2$ to $0.5$, the FWER is (slightly) decreased. In comparison, the number of models $S$ and the parameter $\epsilon$ have a stronger influence on the FWER. In the most extreme case, the FWER is close to $14\%$ for $S=20$ and $\Se_0=\Sp_0=0.9$ for $n=200$ and reaches its target $\alpha=0.025$ not before $n=20000$. Remember that there are only $n_1 = \varrho n$ diseased subjects, so the 'relevant' sample size is much smaller, particularly in the case $\varrho=0.2$. When we relax the least favorable parameter configurations ($\epsilon=0$) to only slightly worse model performances ($\epsilon=0.001$) such as in \eqref{eq:lfc_example}, the FWER declines much faster in $n$. For $\Se_0 = \Sp_0 = 0.8$ the target significance level is met already at around $n=400 \approx \exp(6)$ total or $n_1=80$ diseased observations. For $\Se_0 = \Sp_0 = 0.9$, this number more than doubles to roughly $n=1100\approx\exp(7)$ for $S=10$ and $n=1800\approx\exp(7.5)$ for $S=20$.

In an ancillary simulation ($N_{sim}=5000$), we also investigated results under different true correlation structures of the estimators $\Sehat, \Sphat$ (independence, equicorrelation, autocorrelation) and strengths. None of these changes resulted in qualitatively different results than those shown in figure \ref{fig:lfc_sim}. Additionally, we also investigated the cases $\Se_0=0.8$ and $\Sp_0=0.9$ and vice versa in the LFC scenario ($N_{sim}=5000$). Hereby, it became apparent that the FWER development depends almost solely on the sensitivity value(s), at least in the low prevalence case ($\varrho=0.2$). This can be explained due to the lower number of diseased subjects, and hence lower 'relevant' sample size. When $\varrho=0.5$, the problem becomes symmetric and the FWER curves for $\theta_0=(0.8,0.9)$ and $\theta_0=(0.9,0.8)$ are thus not distinguishable and lie in between those shown in figure \ref{fig:lfc_sim}, i.e. between the cases  $\theta_0=(0.8,0.8)$ and  $\theta_0=(0.9,0.9)$. 

The worst-case projection of parameter values to one, e.g. of $\Sp_1, \Sp_2, \Se_3, \ldots$ in \eqref{eq:lfc_example}, is not really realistic. For instance, in the case $\Se_m=0.9$, $\Sp_m=1$ and $\varrho=0.2$, the overall accuracy is $\Acc_m = 0.98$. In real-world problems, such a high accuracy is rarely attainable. We therefore also considered an alternative data-generating scenario where \eqref{eq:lfc_example} is modified such that the overall accuracy for each model $m$ is bounded by $0.95$. This change decreased the FWER across almost all scenarios by roughly $0.5$ to 2 percent. Otherwise, the results from this sensitivity analysis do not deviate noticeably from those shown in figure \ref{fig:lfc_sim}.

\section{Discussion}\label{sec:discussion}

\subsection{Summary}\label{sec:summary}

In this work we have investigated statistical and study design related methods to improve diagnostic accuracy studies with co-primary endpoints sensitivity and specificity. The main idea is to allow that multiple candidate models are assessed simultaneously on the evaluation data, thereby increasing the likelihood of identifying one that performs sufficiently well. In modern medical testing applications this is necessary because often hundreds or thousands of model architectures and preprocessing pipelines are compared on preliminary data, making a final decision before the evaluation study difficult and prone to error. Our approach allows to wait for the test data until a final decision is made and thereby possibly correcting an erroneous model ranking. 

The main advantages of the multiple testing approach are the on average increased final model performance and the vastly improved statistical power. As multiple comparisons are now conducted on the final dataset, inferences need to be adjusted via a multiple testing procedure. For that matter, we employed a variation of the so-called maxT-approach which is based on a multivariate normal approximation \cite{hothorn2008}. This framework also enables a corrected, conservative (point) estimation but the unbiasedness of the \textit{default} single model evaluation strategy is lost. An important question of practical interest is how models should be selected for evaluation. Besides the heuristic \textit{within\,1\,SE} rule, we also employed the novel \textit{optimal\,EFP} rule. This Bayesian approach aims to optimize the expected final model performance (EFP) before the evaluation study. Given the model assumptions and involved numerical approximations, it provides is with an approximate Bayes action for the considered subset selection problem as it optimizes the posterior expected utility, i.e. the EFP. The underlying model for that matter is based on the hold-out validation data from the model development phase which is readily available before the evaluation study. The comparison regarding all important criteria (performance, power, estimation bias) turns out in favor for the more elaborate \textit{optimal\,EFP} approach. However, the advantage relative to the \textit{within\,1\,SE} rule was small to moderate in most scenarios. For instance, the average performance increase in terms of $\vartheta = \min(\Se, \Sp)$ was around $0.7\%$ in the main simulation study (section \ref{sec:mle_sim}). 
We conclude that the \textit{optimal\,EFP} rule yields a vast improvement over the \textit{default} approach but the conceptually and implementation wise simpler \textit{within\,1\,SE} rule yields almost as good of results. 

As was laid out in section \ref{sec:planning}, an advantage of the \textit{optimal\,EFP} selection rule is that the it does not involve any tuning parameter, in contrast to other potential decision theoretic approaches and the \textit{within\,$k$\,SE} rule ($k>0$). Another major advantage is that the \textit{optimal\,EFP} rule takes into account the evaluation sample size, i.e. it selects more models when more data is (expected to be) available the evaluation study. This is not the case for the simpler \textit{within\,$1$\,SE} rule which only takes into account estimation uncertainty in the validation stage, before the evaluation study. Moreover, we can adapt key characteristics if we know or expect that the learning data is not representative for the evaluation study ahead. This is for instance the case if our preliminary data is from a case-control study with non-representative class balance. As the disease prevalence is a dedicated parameter of our \textit{optimal\,EFP} algorithm (appendix \ref{app:optimal}), we can easily change it to another fixed value or specify a prior distribution. We have not investigated this possible adaptivity in simulations so far. A disadvantage of the \textit{optimal\,EFP} rule is the increased mathematical and numerical complexity.

\subsection{Limitations}\label{sec:limitations}

Our simultaneous test procedure guarantees asymptotic strong control of the Family-wise error rate. As illustrated in section \ref{sec:lfc_sim}, the finite sample performance under least favorable parameter settings can be unsatisfactory, depending on the exact scenario. As expected, the FWER depends primarily on the threshold values $\Se_0$ and $\Sp_0$ because the employed normal approximation performs worse near the boundary of the unit interval. Control of the FWER is also worse when the number of models $S$ is large or when the disease prevalence is low.

We could aim to improve this limitation, e.g. by not utilizing the correlation structure between models which would result in a larger critical value and thus stricter inference. However, as indicated in the main simulation (section \ref{sec:mle_sim}), the average FWER under realistic conditions is far below the significance level, compare figure \ref{fig:mle_sim_fp} and table \ref{tab:mle_sim_table}. Note that in contrast to the simulation study presented in section \ref{sec:lfc_sim}, the analysis in section \ref{sec:mle_sim} is an average risk assessment. The average is taken with respect to different generative distributions over the relevant parameter values, i.e. model performances $\bm \theta =(\Sebm, \Spbm)$ and their dependency structure. These generative distributions are not known but rather implicit in the sense that the relevant parameters dependent on the (feature-label) data generating distribution(s) $\D_{\bm X, Y}$ (which are known) and properties of the learning algorithms. Their resulting characteristics are partially described in figure \ref{fig:mle_sim_eda} and deemed to be realistic for real-world prediction tasks. Note that the results from section \ref{sec:mle_sim} can be generalized to other prediction tasks and learning algorithms which result in similar generative distribution over parameters $\bm \theta$ (and dependency structure thereof).

Together, the results from both simulation studies, realistic scenarios in section \ref{sec:mle_sim} and the worst-case assessment in section \ref{sec:lfc_sim}, indicate that parameter configurations in reality are rarely least favorable in the investigated context. An associated observation was made at the end of section \ref{sec:lfc_sim}. When the LFC is made only slightly less unfavorable, the FWER declines faster in $n$. In the machine learning context, it would be reasonable to assume that the larger the number of models $S$, the less likely least favorable parameters become. After all, the LFC corresponds to the setting that from $M$ initial candidate models, $S$ have the exact same performance at the boundary of the null hypothesis and are all selected for evaluation. While this reasoning does apply to the practitioner, it does not apply from a purely frequentist viewpoint as probability statements over true parameters are irrelevant when (worst-case) FWER control in the strong sense is the goal.

While the multiple testing approaches certainly increase the flexibility in the evaluation study, prediction models still need to be fixed entirely before the study. Verifying this could pose challenges in a regulatory context as merely specifying the learning algorithm and hyperparameters would not be enough in this regard. As the exact resulting models, i.e. their weights, often depend also on the random initialization in the training process, only a full specification of model architecture and weights can be considered sufficient in this regard. Another practical issue is related to the assessment of multiple assisted medical testing systems. Unlike automated disease diagnosis, which is the main target of our methods, a (final) human decision is needed for each prediction in this case. For the evaluation study this would imply, in the worst case, that an equal number of independent readers is needed to supply such a decision. This aspect thus puts a natural (resource) limit on the number of evaluated models and needs to be considered in the planning stage.

\subsection{Extensions}\label{sec:extensions}

Our framework can be extended to classification problems with more than two classes, e.g. different disease types or severity grades. More generally, it can be employed in arbitrary prediction tasks for which not only the overall performance is important but also the performance in different subpopulations. So far, we have not considered such problems in simulation studies. It can however be expected that the operating characteristics of our multiple testing framework will largely depend on the sample size of the smallest subpopulation. This could render the approach practically infeasible for problems with many or an unbalanced distribution of subclasses. 

A rather simple adaptation which might be considered in the future is hypotheses weighting. The employed maxT-approach can rather easily be adopted for that matter \cite{dickhaus2012}. That is to say, rather then splitting up the significance level $\alpha$ equally for each of the $S$ hypotheses, one could spend more on the most promising models. A natural source of information to determine the $\alpha$ ratio is of course the empirical data from the model development phase. Further research on an optimal weighting and its impact on e.g. statistical power is however necessary.

A step further in connecting the evidence before and after the evaluation study would demand Bayesian methods. In a Bayesian model, a prior distribution could be based on validation results and be updated through the evaluation data. We have already experimentally employed a recently proposed multivariate Beta-binomial model for that matter \cite{SIMPle}. While this approach works reasonably well in realistic scenarios, several challenges remain. A main difficulty here is to specify an 'honest' prior distribution, i.e. to not have an arbitrarily optimistic prior. The contrary can however happen in practice when the prior is based on the validation data without any adjustments, in particular when the number of initial candidate models $M$ is large. A correction strategy which accounts for the (subset) selection process before the evaluation study should thus be developed to avoid such overoptimism.

Our \textit{optimal\, EFP} subset selection algorithm is designed to maximize the expected final model performance. While this is a natural goal, there are many other potential utility functions that could reasonably be optimized before the evaluation study. In particular, utility functions based on statistical power, estimation bias or a weighted combination of different criteria appear reasonable. Our novel approach (section \ref{sec:planning}, appendix \ref{app:optimal}) could also be used here to make the specification of the utility function independent of the subset size and thereby simpler for the practitioner. In its current form, our numerical implementation uses several approximations and can potentially be improved regarding its numerical efficiency.




\section*{Acknowledgements}
The first and last author acknowledge funding by the Deutsche Forschungsgemeinschaft (DFG, German Research Foundation) - Project number 281474342/GRK2224/1.

\section*{Conflict of interest}
The authors declare that there is no conflict of interest.


\bibliography{literature}

\clearpage

\appendix

\section{Technical details}\label{app:technical}

\subsection{Parameter estimation}\label{app:estimation}

In the following, we briefly summarize the multivariate Beta-binomial model introduced by \citet{SIMPle} and how it can be used to obtain the regularized point estimators introduced in section \ref{sec:estimation}. We apply the procedure independently for diseased ($\Sebm$) and healthy ($\Spbm$) subsample but keep the notation in the following generic ($\bm \vartheta$) to streamline the presentation. We assume that $n$ observations of $S$ binary random variables $Q_{im}$, $i=1,\ldots,n$, $m\in \setS= \{1,\ldots,S\}$, are available. In the context of this work ${Q_{im}=\one(\hatf_i(\bm X_i)=Y_i)}$ indicates a correct prediction of model $m$ for the $i$-th observation. We suppose independent observations ($Q_{i_1 m}$ and $Q_{i_2 m}$) but two variables $Q_{im_1}$ and $Q_{im_2}$, $m_1,m_2 \in \setS$, may be correlated. This entitles us to assume i.i.d. observations $\bm Q_{i} = (Q_{i1},\ldots,Q_{iS})$ from a multivariate Bernoulli distribution. 
We are interested in the joint distribution of $\bm Q=(Q_1,\ldots,Q_S)$, 
in particular in the marginal means $\vartheta_m = \pr(Q_m=1)$. This may be modeled extensively by a multinomial distribution with parameter vector $\bm p \in \setP = \{(0,1)^{2^S}:\ |\bm p|_1=1\}$ of length $2^S$. Hereby $p_{\bm q} = \pr(\bm Q = \bm q)$ is the probability to observe the event $\bm q \in \{0,1\}^S$. The marginal mean $\vartheta_m$ can then be derived as the sum of all relevant probabilities, i.e. $\vartheta_m = \sum_{\bm q:\, q_m=1} p_{\bm q}$.

A popular Bayesian model that incorporates the multinomial likelihood is the so-called Dirichlet-multinomial model. Hereby, a Dirichlet prior for $\bm p$ with concentration parameter $\bm \gamma \in \mathbb{R}_+^{2^S}$ is assumed. 
According to \citet{SIMPle}, a multivariate Beta distribution can be derived as a simple linear transformation of a Dirichlet distribution and can thus be employed for the situation at hand. 
However, this extensive approach is often infeasible in practice due to large number of parameters.
Besides the full $2^S$-dimensional parametrization, a reduced representation can be derived which consists of a prior sample size $\nu \in \mathbb{R}_+$ and a (symmetric) prior moment matrix $\bm A \in \mathbb{R}^{S\times S}_+$. This amounts to a $\Beta(\bm A_{mm}, \nu-\bm A_{mm})$ prior for $\vartheta_m = \pr(Q_m=1)$ and a $\Beta(\bm A_{m_1 m_2}, \nu-\bm A_{m_1 m_2})$ prior for the probability  $\vartheta_{m_1 m_2}=\pr(Q_{m_1}=1 \wedge Q_{m_2}=1)$. The prior distribution $\pi(\bm \vartheta) \equiv \mBeta(\nu, \bm A)$ can easily be updated via the observed data ($n$, $\bm U$) whereby $n$ is the sample size and the update matrix $\bm U=\bm U (\bm Q)$ consists of elements $\bm U_{m_1 m_2}$ which denote the observed absolute frequencies $\sum_{i=1}^n \one(\ Q_{im_1}=1 \wedge Q_{im_2}=1)$. The posterior distribution of $\bm \vartheta$ given the observed data is then $\pi(\bm \vartheta \given \bm Q) \equiv \mBeta(\nu^\star, \bm A^\star)$ with $\nu^\star=\nu +n$  and $\bm A^\star=\bm A + \bm U$. This simple update rule is similar to the univariate Beta-binomial model and can be derived from the well-studied Dirichlet-multinomial model \cite{SIMPle}. 

We can employ the mBeta-binomial model to replace our usual estimators
\begin{align}\label{eq:prop_est}
\hat{\bm \vartheta} = \bm u / n \quad \text{and} \quad \hat{\bm \Sigma} = (\bm n \bm U - \bm u \bm u\tra)/n^3
\end{align}
for $\bm \vartheta$ and $\bm \Sigma = \cov(\hat{\bm \vartheta})$ whereby $\bm u = \diag(\bm U)$. We will replace the estimators in the last equation by the mean and covariance of the posterior distribution $\pi(\bm \vartheta \given \bm Q)$ which are given by
\begin{align}\label{eq:posterior_params}
\check{\bm \vartheta} = \E \bm \vartheta =\bm \alpha^\star / \nu^\star  \quad \text{and} \quad \check{\bm \Sigma} = \cov(\bm \vartheta) = (\bm A^\star - \bm \alpha^\star (\bm \alpha^\star)\tra) / ((\nu^\star)^2(\nu^\star+1))
\end{align}
whereby $\bm \alpha^\star = \diag(\bm A^\star)$ \cite{SIMPle}. If we employ a vague prior distribution which is composed of $S$ independent uniform distributions this implies $\nu=2$ and all elements of $\bm A$ are equal to $0.5$, except the diagonal which consists of ones. This amounts to adding two pseudo observations (one success, one failure) for each margin whereby half of the success pseudo-observation is counted as common for all variable pairs. Marginally, this approach thus has the same effect as using the popular univariate Beta-binomial posterior mean $\E \vartheta_m = (u_m + 1)/(n+2)$ as a point estimator $\check{\vartheta}_m$ (assuming a uniform prior), compare \citet[Chapter 9]{IBS}. 

\quad

We have applied the estimators $\check{\bm \vartheta}$ and $\check{\bm \Sigma}$ independently for sensitivity and specificity in all numerical experiments in this work. Note that the difference between the usual ($\hat{\bm \vartheta}, \hat{\bm \Sigma}$) and regularized estimators ($\check{\bm \vartheta}, \check{\bm \Sigma}$) vanishes asymptotically. In effect, the asymptotic results in section \ref{sec:inference} apply regardless of which estimators are used. For finite sample sizes the procedure has the advantage to prohibit zero variance estimates and thereby induced singular empirical covariances, compare section \ref{sec:inference}. Strictly speaking, the notation $\mBeta(\nu, \bm A)$ does only define first an (mixed) second-order moments of the distribution and is not a complete characterization. This is however sufficient for our purposes as a multivariate normal approximation is used for the statistical inference which only depends on this information. Note that we also use the mBeta-binomial model for our \textit{optimal\,EFP} selection rule, compare section \ref{sec:planning}. Here it is rather used to construct a generative mBeta distribution from the validation data (via the update rule mentioned above). This generative distribution then builds the foundation to simulate the evaluation study to determine the optimal number of models to evaluate prior to the study, compare also appendix \ref{app:optimal}.

\subsection{Statistical inference}\label{app:inference}

In this section, we re-state and prove the theoretical results from section \ref{sec:inference}. In the following, we work with two samples, i.e. $n_1$ diseased and $n_0$ healthy subjects. Hereby, we need to assume that
$n_1 / n \rightarrow \varrho \notin \{0,1\}$  as $n=n_1+n_0 \rightarrow \infty$. We can then derive $\min(n_1,n_0) \rightarrow \infty$ from $n \rightarrow \infty$. This is for instance given under simple random sampling with a fixed disease prevalence, i.e. when assuming $n_1 \sim \Bin(n, \varrho)$, $n_0=n-n_1$, $\varrho \in (0,1)$.

\begin{theorem}\label{res:distribution}
	Assume $\Se_m -\Se_0 \neq \Sp_m -\Sp_0$ for all $m \in \setS$. Then, under the LFC $\bm \theta^{\bm b} = (\Sebm^{\bm b}, \Spbm^{\bm b})$,
	\begin{align}
	\Tbm^{\bm b}_{(n)} \stackrel{\D}{ \longrightarrow} \setN_S(\bm 0, \bm R^{\bm b}), \quad n \rightarrow \infty,
	\end{align}
	whereby $\Tbm^{\bm b}=\Tbm^{\bm b}_{(n)}$ is defined in \eqref{eq:tstat_dagger} and $\bm R^{\bm b} = \BB \bm R^{\Se} \BB + (\bm I_S - \BB) \bm R^{\Sp}  (\bm I_S - \BB) $.
\end{theorem}

\begin{proof}
	Under the LFC $\bm \theta^{\bb}$, the sub-vector $\TSebm_{\bm b}$ of elements of $\TSebm$ indicated by $b_m=1$ converges in distribution to a multivariate normal random variable of dimension $|\bm b| = \sum_{m=1}^S b_m$ due to the multivariate central limit theorem. Similarly, the sub-vector $\TSpbm_{\bm b}$ of elements of $\TSpbm$ indicated by $b_m=0$ converges to a multivariate standard normal of dimension $S-|\bm b|$. We write
	\begin{align}
	\TSebm_{\bm b} \stackrel{\D}{\longrightarrow} \bm Z^{\Se}_{\bm b} \sim \setN_{|\bm b|}(\bm 0, \bm R^{\Se}_{\bm b}) \quad \text{and} \quad \TSpbm_{\bm b} \stackrel{\D}{\longrightarrow} \bm Z_{\bm b}^{\Sp} \sim \setN_{(S-|\bm b|)}(\bm 0, \bm R^{\Sp}_{\bm b}), \quad n \rightarrow \infty.
	\end{align}
	We denote the full, $S$-dimensional versions of these limiting distributions by $\bm Z^{\Se} \sim \setN_S(\bm \Delta^{\Se}, \bm R^{\Se})$ and $\bm Z^{\Sp} \sim \setN_S(\bm \Delta^{\Sp}, \bm R^{\Sp})$. $\bm Z^{\Se}$ and $\bm Z^{\Sp}$ are singular because they are constant in the remaining, not directly relevant entries. 
	
	The statistics $\TSebm$ and $\TSpbm$ are independent of each other because they are based on distinct samples. This also holds for their limits $\bm Z^{\Se}$ and $\bm Z^{\Sp}$. Moreover, left multiplication with the matrix $({\BB}_2)\tra$ (compare \ref{eq:B2T2}) defines a linear and thus continuous mapping. Together with the continuous mapping theorem, this implies
	\begin{align}
	\Tbm^{\bb}_{(n)} = ({\BB}_2)\tra \bm{T}_{2(n)} \stackrel{\D}{ \longrightarrow} ({\BB}_2)\tra \bm{Z}_2 = \bm Z^{\bb}, \quad n \rightarrow \infty \end{align}  
	whereby $\bm Z_2 \sim \setN_{2S}(\bm \Delta_2, \bm{R}_2)$ with
	\begin{align}
	\bm \Delta_2 = (\bm \Delta^{\Se}, \bm \Delta^{\Sp}) \quad \text{and} \quad \bm{R}_2 = \begin{pmatrix}
	\bm R^{\Se} &\bm 0\\ 
	\bm 0  &\bm R^{\Sp}
	\end{pmatrix}.
	\end{align}
	The matrix $({\BB}_2)\tra$ is fixed and the distribution of $\bm Z^{\bb}$ is thus $\setN_S(\bm 0, \bm R^{\bb})$ with
	\begin{align}
	\bm R^{\bb} = (\BB_2) \bm{R}_2  (\BB_2)\tra = \BB \bm R^{\Se} \BB + (\bm I_S - \BB) \bm R^{\Sp} (\bm I_S - \BB) \in [-1,1]^{S \times S}.
	\end{align}
	Note that $\bm B$ and $(\bm I_S - \BB) $ are diagonal and thus symmetric matrices. This concludes the proof. 
	{\hfill\ensuremath{\square}}
\end{proof}

\begin{lemma}\label{res:consistency}
	Assume $\Se_m -\Se_0 \neq \Sp_m -\Sp_0$ for all $m \in \setS$. Then 
	\begin{align}
	\hat{\bm  R}^{\bm b} = \BBhat \hat{\bm  R}^{\Se} \BBhat + (\bm I_S - \BBhat) \hat{\bm  R}^{\Sp}  (\bm I_S - \BBhat)
	\end{align}
	is a consistent estimator for $\bm R^{\bm b}$.
\end{lemma}

\begin{proof}
	The sample correlation matrix $\hat{\bm R}^{\Se}$ derived from the sample covariance matrix $\hat{\bm \Sigma}^{\Se} $ is a consistent estimator for $\bm R^{\Se}$, the same holds for $\hat{\bm R}^{\Sp}$. Moreover, $\hat{\bb}$ converges in probability to $\bb$. An analogue statement thus holds for $\hat{\BB}$. These individual convergences can be connected with help of the continuous mapping theorem to derive the result.
	{\hfill\ensuremath{\square}}
\end{proof}

\begin{proposition}\label{res:fwer}
	With $c_\alpha$ calculated according to \eqref{eq:fwer_control}, the simultaneous testing procedure $\bm \varphi$ defined by
	\begin{align}
	\varphi_m = 1 \quad  \Longleftrightarrow \quad T_m = \min(T_m^{\Se}, T_m^{\Sp}) > c_\alpha,
	\end{align}
	defines a multiple test with asymptotic strong FWER control at level $\alpha$ for hypothesis system \eqref{eq:hyp}.
\end{proposition}

\begin{proof}
	We replace $\varphi_m$ in the assertion with $\varphi_m^{\bb}$ defined by
	\begin{align}
	\varphi_m^{\bb} = 1 \quad  \Longleftrightarrow \quad T_m^{\bb} > c_\alpha
	\end{align}
	with $T^{\bb}$ defined in \eqref{eq:tstat_dagger}. Then the modified assertion holds due to theorem \ref{res:distribution}, lemma \ref{res:consistency} and the results by \citet{hothorn2008}. As $\bm T^{\bb}$ is unknown (as $\bb$ is unknown) in practice, we need to replace $T^{\bb}$ with $\bm T \leq \bm T^{\bb}$ which will result in equal or less rejections. This proves the original statement.
\end{proof}

\clearpage
\section{Optimal subset selection}\label{app:optimal}

The \textit{optimal\,EFP} model selection rule was described in section \ref{sec:planning}. In this section, the numerical implementation is described in more detail in algorithm \ref{algo_full}. The input parameters are described in the following, including their default values in the simulation study described in section \ref{sec:mle_sim}. 
\begin{itemize}
	\setlength\itemsep{0.075em}
	\item $\Delta_0 = \Se_0-\Sp_0$ indicates the relative importance of sensitivity and specificity. In the numerical experiments in this work, we have only considered the case $\Delta_0=0$.
	\item $\bm Q_\setV = (\bm Q^{\Se}_\setV, \bm Q^{\Sp}_\setV)$ are the validation similarity matrices, compare section \ref{sec:estimation}.
	\item $S_{max}$ is the maximum number of models to be evaluated. By default, we set $S_{max}=\sqrt{n_\setE}$. The main purpose of this threshold is to reduce the computational complexity of the optimization process in our simulation study. 
	\item \texttt{prerank = function($\bm Q_\setV$, $S_{max}$)}. The initial list of candidate models is ranked initially to reduce the subset selection problem $\setS \subset \setM$ to the selection of the optimal number $S \in \{1,\ldots,M\}$ of models to be evaluated. By default, the ranking is conducted according to $\min(T^{\Se}(\setV), T^{\Sp}(\setV))$ whereby the test statistics are calculated similar to \eqref{eq:teststat_def} but based on the validation data $\setV$.
	\item \texttt{moment\_matrix = function($\bm Q$)} converts the similarity matrix $\bm Q$ into the first and second order moment matrix $\bm U$ which was described in appendix \ref{app:estimation}.
	\item $\nu_p=(\nu_p^{\Se},\nu_p^{\Sp} )$ and $\bm A_p = (\bm A_p^{\Se}, \bm A_p^{\Se})$ are initial prior parameters. We employ simple independent uniform priors per default which correspond to $\nu_p=2$ and $A_p$ is a $S \times S$-matrix with all entries equal to $0.5$ except diagonal entries which are equal to $1$. This implies that the generative prior distribution is only determined by the hold-out validation data.
	\item \texttt{max.iter} is the maximum number of iterations, set to 250 by default. This number was mainly chosen for our simulation study. For a single run of the algorithm for a real use-case, a larger number of iterations (500 to 1000) is easily feasible and hence recommended.
	\item \texttt{num.tol} is the numerical tolerance for the stopping criterion, set to $0.001$ by default.
	\item $n$ is the (assumed) sample size for the evaluation study. We assume simple random sampling. Hereby we assume a Beta prior for the disease prevalence $\varrho=\pr(Y=1)$ which is constructed from the learning data. That is, we first sample $\varrho \sim \Beta(1+n_1^\setL, 1+n_0^{\setL})$ and then $n_1 \sim \Bin(n, \varrho)$ and finally set $n_0=n-n_1$. Of course, in practice this could easily be adapted when more information on $\varrho$ is available. This is in particular useful (and recommended) when the learning data is from a case-control study with non-representative disease prevalence.
	\item \texttt{mstar = function($\bm \thetahat$, ...)} is a function of the evaluation data, i.e. of ${\bm \thetahat=(\Sebm, \Spbm)}$ the empirical test performances and potentially further parameters (e.g. $\widehat{\se}(\Sebm)$). By default, the final model is chosen as $m^* =\argmax_{m\in \setS} \min (T_m^{\Se}, T_m^{\Sp})$.
	\item \texttt{Sstar = function($\bm E$, ...)} is a function of the simulation result matrix $\bm E$. By default, we choose $S^*$ as the smallest $S$ such that $\EFPhat(S^*)$ is not more than one standard error smaller than $\max_{m \in \setS}\EFPhat(S)$ in the current iteration. This is illustrated at the end of section \ref{sec:planning} in figure \ref{fig:algo_example}.
	\item \texttt{sample = function(dist, $N$, \ldots)} returns a sample of the input probability distribution \texttt{dist} of size $N$. Samples from the mBeta distribution are obtained via a copula approach \cite{SIMPle}. 
	\item $\pmin(\bm u, \bm v) = (\min(u_m, v_m))_{m=1,\ldots,M}$ is the pairwise minimum of two vectors $\bm u, \bm v$ of length $M$.
\end{itemize}


\SetAlCapSkip{1em}

{
	\begin{algorithm}[H]
		\label{algo_full}
		\SetAlgoLined
		\KwIn{$\bm Q_\setV$, $S_{max}$, $\Delta_0$, $n$, $\nu_p$, $\bm A_p$, \texttt{max.iter}, \texttt{num.tol}}
		\KwResult{$S^*=1$}
		\# rank models and reduce number of models to $S_{max}$\\
		$\bm Q_\setV=$ \texttt{prerank($\bm Q_\setV$, $S_{max}$)}\;
		\# construct posterior (= generative prior) distribution $\pi$ based on prior ($\nu_p, \bm A_p$) and validation data $\bm Q_{\setV}$ \\
		$\nu^{\Se} = \nu_p^{\Se} +$ \texttt{nrow($\bm Q^{\Se}$)}; $\bm A^{\Se}=\bm A^{\Se}_p + \text{\texttt{momemt\_matrix(}} \bm Q^{\Se} \text{\texttt{)}} $; $\pi^{\Se} = \mBeta(\nu^{\Se}, \bm A^{\Se})$\;
		$\nu^{\Sp} = \nu_p^{\Sp} +$ \texttt{nrow($\bm Q^{\Sp}$)}; $\bm A^{\Sp}=\bm A^{\Sp}_p + \text{\texttt{momemt\_matrix(}} \bm Q^{\Sp} \text{\texttt{)}} $; $\pi^{\Sp} = \mBeta(\nu^{\Sp}, \bm A^{\Sp})$\;
		$\pi = (\pi^{\Se}, \pi^{\Sp})$\;
		\# initialize evaluation study result matrix $\bm E$ and $\EFPbmhat$\\
		$\bm E =$ \texttt{matrix(NA,\,nrow=max.iter,\,ncol=$S_{max}$)}; $\EFPbmhat=$\texttt{rep(NA,\,$S_{max}$)}\;
		$i = 1$; $\epsilon = \infty$\;
		\# each iteration of the while loop represents a single simulated evaluation study\\
		\While{\texttt{i} $\leq$ \texttt{max.iter}\ \&\ $\epsilon >$ \texttt{num.tol}}{
			\# sample true parameters $\bm \vartheta = (\Sebm, \Spbm)$ including correlation structure  $\bm C=(\bm C^{\Se}, \bm C^{\Sp})$ from prior $\pi$\\
			$(\Sebm, \bm C^{\Se})=$ \texttt{sample($\pi^{\Se}, N=1$)};
			$(\Spbm, \bm C^{\Sp})=$ \texttt{sample($\pi^{\Sp}, N=1$)};
			$\bm \vartheta = \pmin(\Sebm, \Spbm +\Delta_0)$\;
			\# determine sample sizes $n_1$ and $n_0$ \\
			$\varrho=$ \texttt{sample($\Beta(1+n_1^\setL, 1+n_0^{\setL}), N=1$)}\;
			$n_1 =$ \texttt{sample($\Bin(n, \varrho), N=1$)}; $n_0 = n-n_1$\;
			\# construct sampling distribution of estimators $\Sehatbm, \Sphatbm$ based on these parameters\\
			$\p^{\Se} = \mBin(\Sebm, \bm C^{\Se}, n_1)$;
			$\p^{\Sp} = \mBin(\Spbm, \bm C^{\Sp}, n_0)$\; 
			\# parameter estimates $\hat{\bm \theta}$ from $\p=(p^{\Se}, p^{\Sp})$\\
			$\Sehatbm =$ \texttt{sample($p^{\Se}, N=1$)}/$n_1$;
			$\Sphatbm =$ \texttt{sample($p^{\Sp}, N=1$)}/$n_0$; 
			$\hat{\bm \vartheta} = (\Sehatbm, \Sphatbm)$\;
			\For{S=1 to $S_{max}$}{
				\# determine empirically best model given $S$ where evaluated\\
				$m^* =$ \texttt{mstar($\hat{\bm \theta}$[1:S], ...)}\;
				
				\# determine and save true performance of model $m^*$\\
				$\bm E[i,S]= \bm \vartheta_{m^*} = \pmin(\Se_{m^*}, \Sp_{m^*}+\Delta_0)$\;
			}
			\# estimate EFP by averaging over all simulated evaluation study results so far\\
			$ \EFPbmhat =$ \texttt{colMeans($\bm E$, na.rm=T)}\;
			\# determine optimal number of models\\
			$ S^*=$ \texttt{Sstar($\EFPbmhat$, ...)}\;
			$\epsilon=$ \texttt{sqrt(var($\bm E$[\,,$m^*$], na.rm=T)/i)}; \;
			$i = i+1$\;
		}
		\caption{Pseudo code algorithm for the \textit{optimal\,EFP} model selection rule: Optimization of the expected final model performance depending on the (number of) models to be evaluated.}
	\end{algorithm}
}

\quad

\section{Additional simulation results}\label{app:experiments}

\begin{table}
	\fontsize{11pt}{11pt}\selectfont%
	\rowcolors*{3}{black!10}{}%
	\renewcommand{\arraystretch}{1.25}%
	\arrayrulecolor{black!20}%
	\setlength{\arrayrulewidth}{1pt}%
	
	\caption{Detailed comparison of subset selection rules, stratified for evaluation sample size $n_\setE \in \{400, 800\}$, compare section \ref{sec:mle_sim}. 
	All cells show expectations or probabilities, each estimated based on $N_{sim}=60000$ simulation instances. 
	The last column displays mean differences between the \textit{optimal\,EFP} and \textit{within\,1\,SE} rules including unadjusted 99\% confidence intervals. 
	}
	
	\begin{tabular}{ll|R{1.45cm} R{1.45cm} R{1.9cm} R{2.2cm}|r}
		\toprule
		& $\bm n_\setE$  & \textbf{oracle} & \textbf{default} & \textbf{within\,1\,SE} & \textbf{optimal\,EFP} & \textbf{optimal\,EFP - within\,1\,SE} \\ 
		\midrule\\[-0.6cm]
		\midrule
		$\E\vartheta_*$ & 400  & 0.819 & 0.754 & 0.789 & 0.794 & 0.005 (\phantom{-}0.005, \phantom{-}0.006) \\ 
		& 800  & 0.819 & 0.754 & 0.793 & 0.800 & 0.007 (\phantom{-}0.007, \phantom{-}0.008) \\ 
		$\pr(\vartheta_* > 0.75)$ & 400  & 0.980 & 0.569 & 0.781 & 0.834 & 0.053 (\phantom{-}0.049, \phantom{-}0.057) \\ 
		& 800  & 0.980 & 0.572 & 0.813 & 0.876 & 0.063 (\phantom{-}0.060, \phantom{-}0.067) \\ 
		$rr(0.10)$ & 400  & 0.743 & 0.386 & 0.578 & 0.609 & 0.030 (\phantom{-}0.026, \phantom{-}0.034) \\ 
		& 800  & 0.916 & 0.507 & 0.773 & 0.832 & 0.059   (\phantom{-}0.056, \phantom{-}0.063) \\ 
		$rr(0.05)$ & 400  & 0.236 & 0.085 & 0.105 & 0.111 & 0.007 (\phantom{-}0.004, \phantom{-}0.009) \\ 
		& 800  & 0.514 & 0.180 & 0.280 & 0.301 & 0.021 (\phantom{-}0.017, \phantom{-}0.024) \\  
		$rr(0.00)$ & 400  & 0.004 & 0.001 & 0.001 & 0.001 & 0.000 (-0.000, \phantom{-}0.001) \\ 
		& 800  & 0.007 & 0.001 & 0.001 & 0.001 & -0.000 (-0.000, \phantom{-}0.000) \\ 
		$\FWER$ & 400  & 0.004 & 0.017 & 0.011 & 0.011 & -0.000 (-0.002, \phantom{-}0.001) \\ 
		& 800  & 0.007 & 0.019 & 0.010 & 0.009 & -0.002 (-0.002, -0.001) \\ 
		$\text{Bias}$ & 400  & -0.015 & -0.009 & -0.022 & -0.022 & 0.000 (-0.000, \phantom{-}0.000) \\ 
		& 800 & -0.009 & -0.005 & -0.016 & -0.017 & -0.001 (-0.001, -0.001) \\
		$\MAE_2$ & 400  & 0.034 & 0.076 & 0.053 & 0.046 & -0.006 (-0.007, -0.006) \\ 
		& 800  & 0.026 & 0.071 & 0.044 & 0.037 & -0.008 (-0.008, -0.007) \\ 
		$\pr(O_1)$ & 400  & 0.536 & 0.603 & 0.492 & 0.490 & -0.001 (-0.006, \phantom{-}0.004) \\ 
		& 800  & 0.574 & 0.623 & 0.512 & 0.498 & -0.013 (-0.018, -0.008) \\ 
		$\pr(O_2)$ & 400 & 0.344 & 0.192 & 0.098 & 0.114 & 0.016 (\phantom{-}0.013, \phantom{-}0.019) \\ 
		& 800  & 0.338 & 0.189 & 0.093 & 0.093 & 0.000 (-0.003, \phantom{-}0.003) \\ 
		$\E S$ & 400  & 1.332 & 1.726 & 13.269 & 9.808 & -3.461 (-3.531, -3.392) \\ 
		& 800  & 1.332 & 1.732 & 15.298 & 13.615 & -1.683 (-1.775, -1.592) \\ 
		\bottomrule
	\end{tabular}
	\caption*{\textbf{Performance:}\\
		$\E\vartheta_{*}=\E\min(\Se_*, \Sp_*)$ \quad (expected true performance $\vartheta_{*}=\vartheta_{m^*}$ of final selected model)\\
		$\pr(\vartheta_* > 0.75)$ \quad \hspace{1.5cm} (probability to obtain a final model with performance $\vartheta_*>0.75$)\\
		\textbf{Test decisions:}\\
		$rr(0.10) = \pr(\varphi_{*}=1 \given \delta = 0.10)$ \quad (power, $\delta = \vartheta_{\op} - \vartheta_0$ is the difference between \hfill \\
		 \phantom{-} \hspace{5.5cm} truly highest performance $\vartheta_{\op}$ and threshold $\vartheta_0$)\\
		$rr(0.05) = \pr(\varphi_{*}=1 \given \delta = 0.05)$ \quad (power)\\
		$rr(0.00) = \pr(\varphi_{*}=1 \given \delta = 0.00)$ \quad (unconditional FWER, i.e. for hypothesis system \eqref{eq:hyp_ext})\\
		$\FWER \hspace{0.175cm}  = \pr(\varphi_{*}=1 \given \vartheta_0 = \vartheta_{*})$ \quad \  (conditional FWER, i.e. for hypothesis system \eqref{eq:hyp})\\
		\textbf{Estimation:}\\
		$\text{Bias} \ \ \ \  = \E(\widetilde{\vartheta_*}-\vartheta_*)$  \hspace{3.61cm} (bias regarding estimation of $\vartheta_{*} = \min(\Se_*,\Sp_*)$)\\
		$\MAE_2 = \E(|\widetilde{\Se}_{*}-{\Se}_{*} |+| \widetilde{\Sp}_{*}-{\Sp}_{*} |)/2$ \quad  (combined mean absolute error of corrected point estimates\\ 
		\phantom{-} \hspace{7cm} $\widetilde{\Se}_{*}$ and $\widetilde{\Sp}_{*}$ for parameters $\Se_*$, $\Sp_*$ of final selected model)\\
		$\pr(O_1) = \pr(\widetilde{\Se}_{*} > \Se_{*} \vee\ \widetilde{\Sp}_{*} > \Sp_{*})$ \hspace{1.15cm} (probability to overestimate $\Se_*$ or $\Sp_*$)\\
		$\pr(O_2) = \pr(\widetilde{\Se}_{*} > \Se_{*} \wedge\ \widetilde{\Sp}_{*} > \Sp_{*})$ \hspace{1.15cm} (probability to overestimate $\Se_*$ and $\Sp_*$)\\
		\textbf{Other:}\\
		$\E S$ \quad (expected number of models selected for evaluation study) 
	}
	\label{tab:mle_sim_table}
\end{table}

\end{document}